\PassOptionsToPackage{hyphens}{url}
\documentclass[sigplan]{acmart}
\usepackage{epsfig}
\usepackage{supertech-acm}
\usepackage{floatflt}
\usepackage{array}
\usepackage{multirow}
\usepackage{multicol}
\usepackage{hhline}
\usepackage{dcolumn}
\usepackage{xspace}
\usepackage{epsfig}
\usepackage{wrapfig}
\usepackage{color}

\usepackage{amsmath}
\usepackage{listings}
\usepackage{graphicx}
\usepackage[rightcaption]{sidecap}

\usepackage{flushend}

\IfFileExists{upquote.sty}{\usepackage{upquote}}{}

\providecommand{\color}[2][rgb]{}

\newcommand{\Find}             {\proc{Find}\xspace}
\newcommand{\MakeSet}             {\proc{Make-Set}\xspace}

\newcommand{\creator}             {\texttt{creator}}
\newcommand{\getter}             {\texttt{getter}}
\newcommand{\iack} {\alpha}

\newcommand{\createF} {\texttt{create\_fut}\xspace}
\newcommand{\getF}    {\texttt{get\_fut}\xspace}

\newcommand{\spawn}{\texttt{spawn}\xspace }
\newcommand{\sync}{\texttt{sync}\xspace }
\newcommand{\rlist}{\texttt{reader-list}\xspace }
\newcommand{\lastwriter}{\texttt{last-writer}\xspace }
\newcommand{\spDag}{\text{SP-Dag}}

\newcommand{\full}{\id{full}}
\newcommand{\gfull}{G_{\id{full}}}
\newcommand{\attpred}[1]{#1.\id{attPred}}
\newcommand{\attsucc}[1]{#1.\id{attSucc}}
\newcommand{\find}[1]{\proc{Find}(#1)}
\newcommand{\dssp}{D_{SP}}
\newcommand{\dsnsp}{D_{NSP}}

\newcommand{\FutureRD}{FutureRD\xspace}        
\newcommand{\MultiBags}{MultiBags\xspace}        
\newcommand{\MultiBagsPlus}{MultiBags+\xspace}        

\renewcommand{\floatpagefraction}{1.0}
\usepackage{afterpage}
\usepackage{subcaption}

\makeatletter
\DeclareRobustCommand*\cal{\@fontswitch\relax\mathcal}
\makeatother

\makeatletter
\newtheorem*{rep@theorem}{\rep@title}
\newcommand{\newreptheorem}[2]{%
\newenvironment{rep#1}[1]{%
 \def\rep@title{#2 \ref{##1}}%
 \begin{rep@theorem}}%
 {\end{rep@theorem}}}
\makeatother

\newreptheorem{theorem}{Theorem}
\newreptheorem{lemma}{Lemma}

\startPage{1}

\begin{document}

\title{Efficient Race Detection with Futures}

\author{Robert Utterback}
\orcid{0000-0002-5466-1520}
\affiliation{
\institution{Monmouth College}
}
\email{rutterback@monmouthcollege.edu}

\author{Kunal Agrawal}
\affiliation{
\institution{Washington University in St. Louis}
}
\email{kunal@wustl.edu}

\author{Jeremy Fineman}
\affiliation{
\institution{Georgetown University}
}
\email{jfineman@cs.georgetown.edu}

\author{I-Ting Angelina Lee}
\orcid{0000-0002-0687-5508}
\affiliation{
\institution{Washington University in St. Louis}
}
\email{angelee@wustl.edu}

\copyrightyear{2019} 
\acmYear{2019} 
\setcopyright{acmcopyright}
\acmConference[PPoPP '19]{24th ACM SIGPLAN Symposium on Principles and Practice of Parallel Programming}{February 16--20, 2019}{Washington, DC, USA}
\acmPrice{15.00}
\acmDOI{10.1145/3293883.3295732}
\acmISBN{978-1-4503-6225-2/19/02}

\begin{abstract}

This paper addresses the problem of provably efficient and practically
good on-the-fly determinacy race detection in task parallel programs
that use futures.  Prior works on determinacy race detection have
mostly focused on either task parallel programs that follow a
series-parallel dependence structure or ones with unrestricted use of
futures that generate arbitrary dependences.  In this work, we
consider a restricted use of futures and show that we can detect races
more efficiently than with general use of futures.

Specifically, we present two algorithms: \MultiBags and \MultiBagsPlus.
\MultiBags targets programs that use futures in a restricted fashion and runs
in time $O(T_1 \iack(m,n))$, where $T_1$ is the sequential running time of the
program, $\iack$ is the inverse Ackermann's function, $m$ is the total number
of memory accesses, $n$ is the dynamic count of places at which parallelism is
created.  Since $\iack$ is a very slowly growing function (upper bounded by
$4$ for all practical purposes), it can be treated as a close-to-constant
overhead.  \MultiBagsPlus is an extension of \MultiBags that target programs
with general use of futures.  It runs in time $O((T_1+k^2)\iack(m,n))$
where $T_1$, $\iack$, $m$ and $n$ are defined as before, and $k$ is the number
of future operations in the computation.  We implemented both algorithms and
empirically demonstrate their efficiency.
\end{abstract}

\begingroup
\maketitle
\endgroup

\renewcommand{\shortauthors}{Utterback et al.}
\vspace{-10pt}
\secput{intro}{Introduction}

Races constitute a major source of errors in parallel programs.  Since
they lead to nondeterministic program behaviors, they are extremely
challenging to detect and debug.  In this work, we focus on the problem of
race detection for task-parallel programs, where the programmer denotes the
\emph{logical} parallelism of the computation using high-level parallel
control constructs provided by the platform, and lets the underlying runtime
system perform the necessary scheduling and synchronization.  Examples of task
parallel platforms include OpenMP~\cite{OpenMP13}, Intel's
TBB~\cite{Reinders07, IntelTBBManual}, IBM's X10~\cite{CharlesGrSa+05},
various Cilk dialects~\cite{FrigoLeRa98, DanaherLeLe06, Leiserson10,
IntelCilkPlusLangSpec13}, and Habanero dialects~\cite{BarikBuCa09,
CaveZhSh11}.

In the context of task parallel programs, the focus is typically on detecting
\defn{determinacy races}~\cite{FengLe97} (also called \defn{general
races}~\cite{NetzerMi92}), which occur when two or more logically parallel
instructions access the same memory location and at least one access is a
write.  In the absence of a determinacy race, a task parallel program for a
given input behaves \emph{deterministically}.

Over the years, researchers have proposed several determinacy race 
algorithms~\cite{Mellor-Crummey91, FengLe97, FengLe99, RamanZhSa10,
RamanZhSa12, BenderFiGi04, Fineman05, UtterbackAgFi16, SurendranSa16,
XuLeAg18} for task parallel code.  These algorithms perform race detection
\defn{on the fly} as the program executes, and consist of two main components:
(1) an \defn{access history} that keeps track of previous readers and writers
for each memory location; and (2) a \defn{reachability data structure} for
maintaining and querying whether two instructions are logically in parallel.  
On each memory access, the detector checks whether the current access is
logically parallel with the previous accessors (stored in the access history)
to determine whether a race exists.

Most prior work focuses on a restricted set of computations, namely
computations that can be represented as \defn{series-parallel dags (SP
dags)}~\cite{Valdes78} with nice structural properties, such as ones
generated using fork-join parallelism (i.e.,\spawn/\sync or
\texttt{async}/\texttt{finish}).  Prior works show that one can race detect
computations that are SP dags efficiently by exploiting the nice
structural properties.  In particular, the reachability data structure
can be maintained and queried with \emph{no asymptotic overhead} for
both serial~\cite{BenderFiGi04, Fineman05} and parallel
executions~\cite{UtterbackAgFi16}.  Moreover, the access history needs
to store only a constant number of accessors per memory location to
correctly race detect for such computations~\cite{FengLe97, FengLe99,
Mellor-Crummey91}.

The use of \defn{futures} has become a popular way to extend fork-join
parallelism.   Since their proposal~\cite{FriedmanWi78, BakerHe77} in the late
70s, futures have has been incorporated into various parallel
platforms~\cite{LuJiSc14, CaveZhSh11, FluetRaRe10, CharlesGrSa+05,
ChandraGuHe94, KranzHaMo89, ArvindNiPi86, Halstead85}.  Researchers have studied
scheduling bounds~\cite{BlellochGiMa97, AroraBlPl98} and cache
efficiency~\cite{SpoonhowerBlGi09,HerlihyLi14} for using futures with
fork-join computations.  \citet{KoganHe14} study linearizability of concurrent
data structures accessed using futures.  \citet{SurendranSa16b} proposed using
futures to automatically parallelize programs.

The use of futures can form arbitrary dependencies, and thus computations
generated by a parallel program that uses futures are no longer
series-parallel.  However, not much work has been done on race detecting
programs with more general dependence structures. 

Two prior works exist on race detection for programs that use futures and
both are sequential (no known parallel algorithms exist).  An algorithm
proposed by~\citet{SurendranSa16a} has high overheads --- the running time is
$O(T_1(f+1)(k+1))$ where $T_1$ is the \defn{work}, or sequential running time
of the program without race detection, $f$ is number of future objects and $k$
is the number of future operations.  That is, the running time of the race
detection algorithm increases quadratically with the total number of futures
used in the program.
More recently, \citet{AgrawalDeFi18} present a sequential algorithm to perform race
detection on SP dags with $k$ added non-series-parallel edges in $O(T_1+k^2)$
time, which is the best known running time.  The algorithm is difficult to
implement however, since it requires storing all the nodes in the computation
graph and traversing the graph during execution to update labels.  Thus, no
actual implementation of the algorithm exists to date.

\subsection*{Contributions}

While prior work on race detection has focused on either structured SP dags or
unrestricted use of futures that generates arbitrary dependences, we consider
a restricted use of futures.  Researchers have observed in other
contexts~\cite{HerlihyLi14} that using futures in a restricted manner can
reduce scheduling and cache overheads.  We define a specific \defn{structured}
use-case of futures that allows us to perform race detection much more
efficiently than general use of futures.  This class of futures is quite
natural and can be checked with program analysis.  We provide the precise
definition in \secref{prelim}; informally, it requires that the instruction
that creates the future is sequentially before the instruction that uses the
handle.

We present two practical algorithms for race detecting programs with futures:
\MultiBags and \MultiBagsPlus.  The main contribution for both algorithms is a
novel reachability data structure.  Both algorithms run the program
sequentially for a given input and report a race if and only if one exists,
following the same correctness criteria as prior work.  \MultiBags focuses on
structured use of futures and incurs very little overhead --- a multiplicative
overhead in the inverse Ackermann's function, which is upper bounded by $4$
for all practical purposes~\cite{CormenLeRi09}.  \MultiBagsPlus is an
extension of \MultiBags, which handles general use of futures and has overhead
comparable to the state-of-the-art theoretical algorithm~\cite{AgrawalDeFi18}
(i.e., multiplicative overhead of the inverse Ackermann's function) and can be
implemented efficiently.  We have implemented both algorithms and empirically
evaluated them.  The empirical results show that both algorithms can maintain
reachability efficiently for their designated use cases.

Specifically, we make the following contributions:
\begin{closeitemize}

\item \textbf{\MultiBags:} We propose \MultiBags, an algorithm to race detect
programs that use structured futures (\secref{structured}).  We prove its
correctness and show that it race detects in $O(T_1\iack(m,n))$ time where
$T_1$ is the work, $\iack$ is the inverse Ackermann's function (upper bounded
by $4$), $m$ is the number of memory accesses, and $n$ is the dynamic count of
places at which parallelism is created.  Since $\iack$ is a very slow-growing
function, this bound is essentially $O(T_1)$ for all intents and purposes.

\item \textbf{\MultiBagsPlus:} We propose \MultiBagsPlus, an algorithm to race
detect programs that use general futures (\secref{general}).  We prove its
correctness and show that it race detects in $O(T_1 + k^2)\iack(m,n)$ time
where $T_1$, $\iack$, $m$, and $n$ are defined as above, and $k$ is the number
of future operations in the computation (\secref{general}).  Again, since the
inverse Ackermann's function is slow growing function, the running time is
$O(T_1 + k^2)$ for all intents and purposes.  Compared to the state-of-the-art
proposed by~\citet{AgrawalDeFi18}, \MultiBagsPlus's running time has a
multiplicative overhead of the inverse Ackermann's function.  Unlike the
state-of-the-art, however, \MultiBagsPlus's relative simplicity allows it to
be implemented efficiently in practice.  We provide a more detailed comparison
between our \MultiBagsPlus algorithm and the state-of-the-art~\cite{AgrawalDeFi18}
in \secref{general}.

\item \textbf{\FutureRD:} We have built a prototype race detector called
\FutureRD based on \MultiBags and \MultiBagsPlus.  Empirical evaluation with
\FutureRD shows that our algorithms allow reachability to be maintained
efficiently, incurring almost no overhead (geometric means of
$1.06\times$ and $1.40\times$ overhead for \MultiBags
and \MultiBagsPlus, respectively).  The overall race detection incurs
geometric means of $20.48\times$ and $25.98\times$ overhead,
respectively.

\end{closeitemize}

\secput{prelim}{Preliminaries and Definitions}


\paragraph{Parallel control constructs:} 
Our algorithms are described assuming parallelism in programs is
generated using four primitives: \spawn, \sync, \createF and \getF.
The algorithms themselves are general and can be applied to
platforms that use other constructs that generate similar types of
dags.  We assume that \spawn and \sync are used to generate fork/join
or series/parallel structures.  In particular, for the purposes of
this paper, function $F$ can \defn{spawn} off a child function $G$,
invoking $G$ without suspending the continuation of $F$, thereby
creating parallelism; similarly, $F$ can invoke \defn{sync}, joining
together all previously spawned children within the functional
scope.\footnote{Some constructs, such as async/finish primitives have
  slightly different restrictions, they still generate SP dags and our
  algorithms can be modified to apply to these programs.}  We assume
\defn{create\_fut} and \defn{get\_fut} primitives are used to
create and join futures, respectively.  Like \spawn, one can precede a
function call to $G$ in $F$ with \createF, which allows $G$ to execute
without suspending $F$.  Unlike \spawn, however, parallel function
calls created with \createF\ can escape the scope of a \sync\ --- a
subsequent \sync joins together previously spawned functions but does
not wait for function calls preceded by \createF to return.  Instead,
\createF returns a \defn{future handle} $h$, which the program must
explicitly invoke \getF on to join with the corresponding computation
(i.e., $G$). If $G$ has not completed when \getF\ is called, then
\getF\ \defn{blocks} until $G$ finishes and a result is obtained.

\paragraph{Modeling parallel computations:}
One can model the execution of a parallel program for a given input 
as a \defn{dag (directed acyclic graph)} $\gfull$, whose nodes are
\defn{strands} --- sequence of instructions containing no parallel
control -- and edges are control dependencies among
strands.  The dag unfolds dynamically as the program executes.  
A strand $u$ is sequentially before another strand $v$ (denoted by
$u \prec v$) if there is a path from $u$ to $v$ in the dag; two nodes
$u$ and $v$ are \defn{logically parallel} if there is no path from one
to the other.  
The performance of a computation can be measured in two terms: the \defn{work}
$T_1$ of the computation is the execution time of the computation on a single
processor; the \defn{span} (also called \defn{depth} or 
\defn{critical-path length}) $T_\infty$ of the computation is its execution 
time on an infinite number of processors (or, longest sequential path through
the dag).  

\paragraph{Series-parallel dags:} Computations which use only \spawn
and \sync can be modeled as \defn{series-parallel dags
  (SP-dag)}~\cite{Valdes78} that have a single \defn{source} node with
  no incoming edges and a single \defn{sink} node with no out-going
  edges.  Upon the execution of a \spawn, a \defn{fork node} is
  created with two outgoing edges: one leads to the first strand in
  the spawned child function and one leads to the continuation of the
  parent.  Upon the execution of a \sync, a \defn{join node} is
  created, that has two or more incoming edges, joining the previously
  spawned subcomputations.\footnote{Technically, a \sync can join
  multiple children; therefore, a join node can have more than two
  parents.  For simplicity, in this paper, we assume each join has
  exactly two incoming edges. One can modify all our algorithms to the
  more general case easily.}

SP dags can be constructed recursively as follows.
\begin{closeitemize} 
\item \defn{Base Case}: the dag consists of a single node that is
    both the source and the sink.

  \item \defn{Series Composition}: let $G_1 = (V_1,E_1)$ and
    $G_2 = (V_2,E_2)$ be SP-dags on distinct nodes.  Then a series
    composition $G$ is formed by adding an edge from $\id{sink}(G_1)$
    to $\id{source}(G_2)$ with $\id{source}(G) = \id{source}(G_1)$ and
    $\id{sink}(G) = \id{sink}(G_2)$.

  \item \defn{Parallel Composition}: let $G_L = (V_L,E_L)$ and
    $G_R = (V_R,E_R)$ be SP-dags on distinct nodes.  Then the parallel
    composition $G$ is formed as follows: add a fork node $f$ with
    edges from $f$ to both sources, and a join node $j$ with edges
    from both sinks to $j$.  $\id{source}(G) = f$ and
    $\id{sink}(G) = j$.  We refer to $G_L$ and $G_R$ as the \defn{left
      subdag} and \defn{right subdag}, respectively, of both the fork
    $f$ and join $j$.
\end{closeitemize}

\paragraph{Adding futures:} We model computations that employ
futures in addition to \spawn and \sync as a set of independent SP dags 
connected to each other via \defn{non-SP edges} due to \createF and
\getF calls.  If a function $F$ spawns a function $G$, then the strands
of $F$ and $G$ are part of the same SP dag.  However, if function $G$
calls $H$ using a $\createF$ call, then the first strand, say $v$, of
$H$ is the source of a different SP dag.  The last strand of $H$ will
be the sink node of this SP dag.  Therefore, if a program calls
\createF $f$ times --- that is, it creates $f$ total futures in addition 
to the main program --- then it has $f+1$ SP dags which are connected
to each other via non-SP edges.

These non-SP edges are incident on strands that end with \createF and
ones that immediately follow strands that ended with \getF.  A strand
$u$ in function $G$ that ends with $h=\createF(F)$ has two outgoing
edges --- one non-SP edge to the first strand in $F$, and one SP edge
to the continuation in $G$.  We say that $u$ is the \defn{creator} of
$F$ denoted by $\creator(F_j)$.  The first strand of $F$ is the source
of a new SP-dag which contains all strands of $F$ and the functions it
calls (recursively) using \spawn and the last strand of $F$ is the
sink of this SP-dag.  Similarly, strand $u$ in $H$ immediately follows
a $\getF(h)$ call where $h$ is the future handle for future $F$ has
two incoming edges --- one SP edge from the strand that ended with
the \getF call in the current function $H$ and one non-SP edge from
the last strand of the $F$.  We say $u$ is the \defn{getter} of $F_j$
denoted as $\getter(F)$.  We say that $u \prec_{SP} v$ if there is a
path from $u$ to $v$ using only SP edges.

This model is quite general and subsumes computations that can arise from
futures~\cite{LuJiSc14, CaveZhSh11, FluetRaRe10, CharlesGrSa+05,
  ChandraGuHe94, KranzHaMo89, ArvindNiPi86, Halstead85} or other
future-like (such as ``put'' and ``get''~\cite{BudimlicBuCa10,
  TasirlarSa11}) parallel constructs proposed in the literature.  
Therefore, our algorithm would work on all of these primitives.


\paragraph{Structured futures:} We place the following restrictions on 
structured futures: (1) \defn{Single-touch:} Every future handle
is called with \getF at most once. (2)
\defn{No race on future handles:} There is a sequential dependence
in the program from the point where a future is created (via \createF which
initializes a future handle) to the point where it is read (via \getF).
More precisely, if strand $u$ terminates with a $f \gets
\createF(F)$ call and strand $v$ terminates with a $\getF(f)$ call, then
$u\prec v$ in the computation.  


\paragraph{Eager execution:} Both our algorithms execute the
computation sequentially and execute the program in \defn{depth-first
  eager} order.  When the execution reaches $\createF(F)$
(after executing $\creator(F)$) call or a $\spawn(F)$ call, it always
executes the function $F$.  When $F$ returns, then the next node of
the parent function (the one after the continuation edge) is executed.
This execution order automatically has the property that all functions
that must join at a \sync point have already returned when the
execution reaches the \sync; therefore, the execution never blocks at
a \sync.  Similarly, for structured futures, this execution has the
property that the execution will never block at a \getF.  For general
futures, we restrict our attention to computations where the use of
futures is \defn{forward-pointing}: for every future $F$,
$\creator(F)$ executes before $\getter(F)$ in the depth-first eager
execution.  Without this restriction, sequential execution of the
original program could deadlock, in which case our algorithm race
detects up to the point where it deadlocks.



\secput{race}{Managing Access History}

As mentioned in \secref{intro}, there are two important components in
a race detector: access history and reachability data structure.
\MultiBags and \MultiBagsPlus differ in how they maintain
reachability but manage access history similarly.  This section
discusses how they manage access history --- for each memory location
$\ell$, the access history maintains enough information about the
previous accesses to $\ell$ so that future accesses to $\ell$ can
detect races.

When race detecting a series parallel program, it is sufficient to store a
constant number of previous reader strands and a single previous writer strand
in the access history\cite{FengLe97, Mellor-Crummey91}.  When a strand $s$
accesses a memory location $\ell$, it checks if some subset (based on whether
$s$ is reading or writing) of $\ell$'s previous accessors are in parallel with
$s$.  Therefore, each memory access leads to at most a constant number of
queries into the reachability data structure.
 

This property no longer holds for programs with futures, however.  In
particular, the access history for a memory location $\ell$ still holds
only one writer strand, namely the most recent writer strand, 
$\lastwriter(\ell)$.  However, it must now store an arbitrarily large
$\rlist$.  Race detection proceeds as follows.  Whenever a strand $s$
reads from a memory location $\ell$, the detector checks the
reachability data structure to determine whether $s$ is logically
parallel with $\lastwriter(\ell)$; if so, a race is reported.
Otherwise, $s$ is added to $\rlist(\ell)$.  When a strand $s$ writes
to a memory location $\ell$, the race detector must check $s$ against
all readers in $\rlist(\ell)$ and with $\lastwriter(\ell)$.  If $s$ is
in parallel with any of them, then it declares a race.  Otherwise, the
$\rlist(\ell)$ is set to $\emptyset$ and $s$ is stored as
$\lastwriter(\ell)$.  We can empty the reader list without missing any
races because anything that executes later that would be in parallel
with these readers must also be in parallel with $s$ (the new
$\lastwriter(\ell)$), and a race will be reported with $s$.

A key thing to notice here is the following: the total number of
queries into the reachability data structure (i.e., checking one
access against another for race) is bounded by the total number of
memory accesses in the computation.  Since we can
empty all the readers $\rlist(\ell)$ whenever we encounter a
$\lastwriter(\ell)$, a reader is checked against some writer at most
twice: when it is inserted into the reader list, and right before it 
is removed.  Thus, the total number of queries to the reachability 
data structure is bounded by $O(T_1)$, where $T_1$ is the work of 
the computation.  We shall relate
this observation formally to the performance bound of \MultiBags and
\MultiBagsPlus later in \secreftwo{structuredOmitted}{generalOmitted}.


\secput{structured}{\MultiBags for Structured Futures}

We now describe \MultiBags, which can race-detect programs with
structured futures in time $O(T_1 \iack(m, n))$ where $T_1$ is the
work of the program, $\iack$ is the inverse Ackermann's function, $m$
is the number of memory accesses in the program and $n$ is the number
of \spawn and \createF calls.  Since the inverse Ackermann's function is 
a very slowly growing function, the bound is close to optimal.

\paragraph{Notation:}
Note that programmatically, \spawn and \sync are subsumed by \createF
and \getF since we can convert a \spawn to \createF and \sync to a
series of \getF calls, one on each function spawned in the current
function scope.  In the case of general use of futures discussed in
\secref{general}, we distinguish between SP edges (generated
by \spawn and \sync) and non-SP edges generated by \createF and \getF
since the bound depends on $k$, the number of \getF calls, and
converting all \sync calls to \getF calls will increase this number.
For structured futures, however, the bound does not depend on $k$;
therefore, for simplicity in this section, we assume that we only
have \createF and \getF constructs to create parallelism.

The computation dag consists of three kinds of nodes --- regular
strands with one incoming and one outgoing edge, \defn{creator}
strands which end with a \createF call with two outgoing edges, and
\defn{getter} strands (that come immediately after a \getF call) with
two incoming edges.  It also consists of three kinds of edges:
\defn{spawn} edges are edges from creator nodes to the first strand of
the future; \defn{join} edges are edges from last strand of a future
to getter nodes; all other edges (that go between strands of the same
function instance) are \defn{continue}
edges. 

\subsection{Algorithm}

This algorithm is similar to the SP-Bags algorithm for detecting races
for series-parallel programs~\cite{FengLe97}.  As with that algorithm,
we will use a the fast disjoint-set data
structure~\cite{Tarjan75}. The data structure maintains a dynamic
collection $D$ of disjoint sets and provides three operations:
\begin{closeitemize}
\item $A = \proc{Make-Set}(D,x):$ Creates a new set $A = \set{x}$ and
  adds it to the disjoint sets data structure $S$.
\item $A = \proc{Union}(D, A, B):$ Unions the set $B$ into $A$ and
  destroys $B$. We will sometimes overload notation and say
  $\proc{Union}(D, x,y)$ where $x$ and $y$ are elements in the set
  instead of sets.  This means that we union the sets containing $x$
  and $y$ into the set containing $x$. 
\item $\proc{Find}(D, x)$ returns the set that contains the element $x$.
\end{closeitemize}
In this section, we only have one disjoint set data structure;
therefore, $D$ is implicit.

As mentioned in \secref{prelim}, and like SP-Bags, \MultiBags depends
on the depth-first eager execution of the computation.  
\MultiBags maintains a bag (a set in the union-find data structure)
for each function instance $F$ which has been created and for which
$\getF$ has not yet been called (these bags can be stored with the
future handle).  This bag is labeled either an $S$-bag, represented by
$S_F$ or a $P$-bag, represented by $P_F$.  The algorithm maintains
these bags as shown in Figure~\ref{code:struct-alg}.  The strands of a
particular function $F$ are always added to $S_F$ before they execute.

\begin{figure}
  \footnotesize
  \begin{center}
  \begin{tabular}{|l|}
    \hline
    \begin{minipage}[t]{\columnwidth}
      $F$ calls $f=\createF(G)$ where $u$ is the first strand of $G$: 
            \vspace{-1ex}
      \begin{codebox}
        \li $S_G \gets \proc{Make-Set}(u)$ 
      \end{codebox}
      \vspace{-1ex}
      $G$ returns to $F$:\vspace{-1ex}
      \begin{codebox*}
        \li $P_G \gets S_G$.
      \end{codebox*}
      \vspace{-1ex}
      $F$ calls $y \gets \getF(f)$ where $f$ is $G$'s handle:\vspace{-1ex}
      \begin{codebox*}
        \li $S_{F} \gets \proc{Union}(S_{F}, P_G)$ \lilabel{union}
      \end{codebox*}
    \end{minipage}%
    \\ \hline 
    \begin{minipage}[t]{\columnwidth}
      \Comment{Called when strand $v$ accesses memory location $\ell$}\\
      \Comment{previously accessed by $u$ in a conflicting way}:\vspace{-1ex}
      \begin{codebox*}
        \Procname{$\proc{Query}(u,v)$ 
            \Comment{return \const{true} iff $u\prec v$}}
          \li \If $\proc{Find}(u)$ is an $S$ bag, \Return \const{true}
          \li \Else \Return \const{false}
            \vspace{-1ex}
      \end{codebox*}
    \end{minipage}\\
    \hline
  \end{tabular}
  \end{center}
  \vspace{-2mm}
  \caption{Pseudocode for \MultiBags.  The top part shows how
    \MultiBags maintains the $S$ and $P$ bags when it encounters
    future constructs.  The bottom part shows the operation of 
    checking for races upon a memory accessed.}
\vspace{-2mm}
\label{code:struct-alg}
\end{figure}

\punt{
\begin{figure}
  \small
  \begin{minipage}[t]{0.5\linewidth}
    \begin{codebox}
      When $F$ calls $f=\createF(G)$:\vspace{-.5ex}
      \li $S_G \gets \proc{Make-Set}(u)$ where $u$ is the first strand of $G$
      When $G$ returns to $F$:\vspace{-.5ex}
      \li $P_G \gets S_G$.
      When $F$ calls $y \gets \getF(f)$ where $f$ is $G$'s handle.
      \li $S_{F} \gets \proc{Union}(S_{F}, P_G)$ \lilabel{union}
    \end{codebox}
  \end{minipage}%
  \begin{minipage}[t]{0.47\linewidth}
    \begin{codebox}
       \Procname{$\proc{Query}(u,v)$}
         \zi \Comment Called when stand $v$ accesses memory location $\ell$ 
         \zi \Comment return \const{true} iff $u\prec v$
         \li \If $\proc{Find}(u)$ is an $S$ bag, \Return \const{true}
         \li \Else \Return \const{false}
    \end{codebox}
  \end{minipage}
\label{code:struct-alg}
\end{figure}

\begin{figure}
  \small
  \begin{subfigs}{ll} 
    \toprule
    \begin{minipage}[t]{0.5\linewidth}
    When $F$ calls $f=\createF(G)$:\vspace{-.5ex}
    \begin{codebox}
      \li $S_G \gets \proc{Make-Set}(u)$ where $u$ is the first strand of $G$
    \end{codebox}
    \end{minipage}
    &
    \begin{minipage}[t]{0.47\linewidth}
    When $G$ returns to $F$:\vspace{-.5ex}
    \begin{codebox}
      \li $P_G \gets S_G$.
    \end{codebox}
    \end{minipage}
    \\
    \midrule
    \begin{minipage}[t]{0.5\linewidth}
      When $F$ calls $y \gets \getF(f)$ where $f$ is $G$'s handle.
      \begin{codebox}
          \li $S_{F} \gets \proc{Union}(S_{F}, P_G)$ \lilabel{union}
      \end{codebox}
    \end{minipage}%
    &
    \begin{minipage}[t]{0.47\linewidth}
      \begin{codebox}
        \Procname{$\proc{Query}(u,v)$}
          \zi \Comment Called when stand $v$ accesses memory location $\ell$ 
          \zi \Comment return \const{true} iff $u\prec v$
          \li \If $\proc{Find}(u)$ is an $S$ bag, \Return \const{true}
          \li \Else \Return \const{false}
      \end{codebox}
    \end{minipage}
    \\
    \bottomrule
  \end{subfigs}
\label{code:struct-alg}
\end{figure}
}

\punt {
\begin{figure}
\begin{codebox}
\li Function $F$ calls $f=\createF(G)$
\li $u$ is the first strand of $G$
\li $S_G \gets \proc{Make-Set}(u)$
\zi
\li Function $G$ returns
\li $P_G \gets S_G$.
\zi
\li Function $F$ calls $y \gets \getF(f)$ where $f$ is $G$'s handle.
\li $S_{F} \gets \proc{Union}(S_{F}, P_G)$ \lilabel{union}
\end{codebox}
\caption{Pseudocode for maintaining bags for structured futures.}
\label{code:struct-alg}
\end{figure}

\begin{figure}
\begin{codebox}
\Procname{$\proc{Query}(u,v)$}
\zi \Comment Called when stand $v$ accesses memory location $\ell$ 
\zi \Comment return \const{true} iff $u\prec v$
\li \If $\proc{Find}(u)$ is an $S$ bag, \Return \const{true}
\li \Else \Return \const{false}
\end{codebox}
\caption{Query in structured futures}
\label{code:struct-query}
\end{figure}
}

\begin{figure}
  \begin{tabular}{c}
    \begin{minipage}{1.75in}
      \includegraphics[width=1.75in]{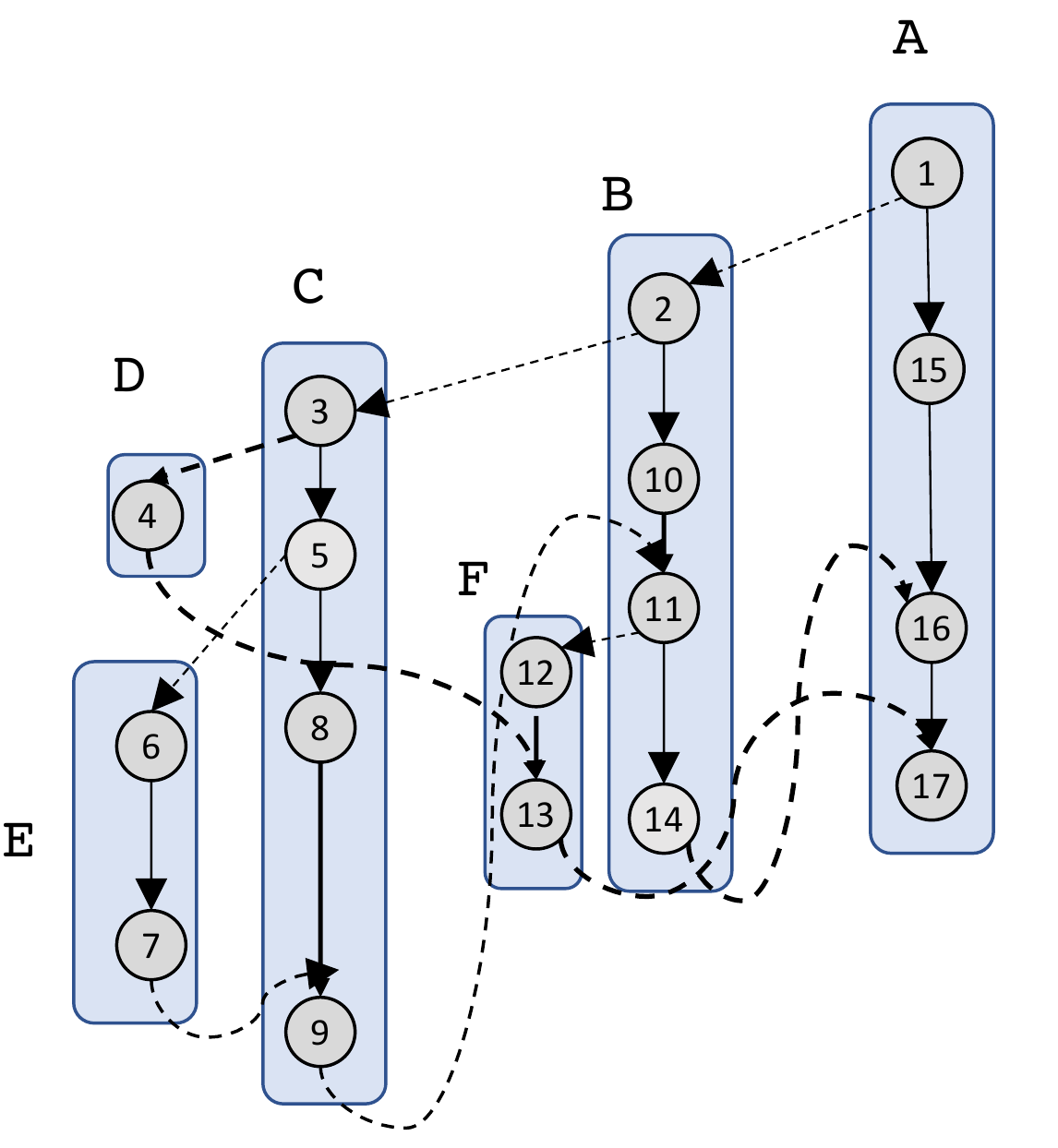}
    \end{minipage}
    \\ 
    \begin{minipage}{3in}
    {\scriptsize     \begin{tabular}{|c|c|}\hline
      node &  $D$ \\ \hline\hline
      1  & $S_A = \set{1}$ \\
      2  & $S_A = \set{1}$, $S_B=\set{2}$ \\
      3  &$S_A = \set{1}$, $S_B=\set{2}$, $S_C = \set{3}$ \\
      4  & $S_A = \set{1}$, $S_B=\set{2}$, $S_C = \set{3}$,$S_D = \set{4}$ \\
      5  & $S_A = \set{1}$, $S_B=\set{2}$, $S_C = \set{3,5}$,$P_D = \set{4}$ \\
      6  & $S_A = \set{1}$, $S_B=\set{2}$, $S_C = \set{3,5}$,$P_D = \set{4}$, $S_E=\set{6}$ \\
      7  & $S_A = \set{1}$, $S_B=\set{2}$, $S_C = \set{3,5}$,$P_D = \set{4}$, $S_E=\set{6,7}$ \\
      8  & $S_A = \set{1}$, $S_B=\set{2}$, $S_C = \set{3,5,8}$,$P_D = \set{4}$, $P_E=\set{6,7}$\\
      9  & $S_A = \set{1}$, $S_B=\set{2}$, $S_C = \set{3,5,6,7, 8, 9}$,$P_D = \set{4}$ \\
      10 & $S_A = \set{1}$, $S_B=\set{2,10}$, $P_C = \set{3,5,6,7, 8, 9}$,$P_D = \set{4}$ \\
      11 & $S_A = \set{1}$, $S_B=\set{2,3,5,6,7,8, 9, 10,11}$,$P_D = \set{4}$ \\
      12  & $S_A = \set{1}$, $S_B=\set{2,3,5,6,7, 8,9,10,11}$,$P_D = \set{4}$, $S_F=\set{12}$ \\
      13 & $S_A = \set{1}$, $S_B=\set{2,3,5,6,7, 8,9,10,11}$, $S_F=\set{4, 12,13}$ \\
      14 & $S_A = \set{1}$, $S_B=\set{2,3,5,6,7, 8,9,10,11, 14}$, $P_F=\set{4, 12,13}$ \\
      15 &  $S_A = \set{1, 15}$, $P_B=\set{2,3,5,6,7, 8,9,10,11, 14}$, $P_F=\set{4, 12,13}$ \\
      16 & $S_A = \set{1,2,3,5,6,7, 8,9,10,11, 14,15,16}$, $P_F=\set{4, 12,13}$ \\
      17 & $S_A = \set{1,2,3,4, 5,6,7, 8,9,10,11, 12,13,14,15, 16,17}$ \\
\hline
    \end{tabular}
}
  \end{minipage}
  \end{tabular}
  \vspace{-2mm}
  \caption{An example execution of \MultiBags on a program with
    structured use of futures.  In this program, there is always a
    sequential dependence between each future's creator and its
    corresponding getter's immediate predecessor in the same function 
    (e.g., $B$'s creator and getter are 1 and 16 respectively, and
    the immediate predecessor of the getter in the same function is 15). 
    This dag is not a series-parallel dag, as the spawning and joining
    of function instances are not well nested.  The table shows the 
    state of the disjoint-set data structure for maintaining 
    reachability immediately before the execution of each strand in 
    the order of the execution.}\figlabel{struct-ex}
\end{figure}

The algorithm looks similar to SP-Bags~\cite{FengLe97}. The main
difference is that when the function $G$ returns, its $S$-bag $S_G$ is
renamed as $P_G$ bag; in SP-bags, $S_G$ would be unioned with $P_F$,
the parent function of $G$. However, this small difference is crucial
for handling programs with structured futures rather than
series-parallel programs.  In addition, the proof for correctness of
this algorithm for structured futures is significantly different than
the proof of correctness of SP-bags for series-parallel programs.

Figure~\ref{fig:struct-ex} shows the operation of this algorithm on an
example program which uses structured futures.  Each rectangle is a
function instance and nodes are strands.  The straight dashed lines going
towards the left represent \createF edges while the curved dashed
lines represent \getF edges.  Consider step 12 when the first node of
function $F$ is executing.  All nodes except node 4 are sequentially
before this strand and are correspondingly in some $S$-bag.  Node 4 is
in parallel with this strand and is in a $P$-bag.

\subsection{Proof of Performance and Correctness}

First, we state the performance bound --- the proof is essentially
identical to the proof of performance of SP-Bags, since the algorithm
is quite similar.

\subheading{Proof of Performance of MultiBags}

\begin{theorem}
\thmlabel{performancestruct}
The running time of \MultiBags when detecting races for a program with
work $T_1$ is $T_1\iack(m,n)$ where $m$ is the number of memory accesses
and $n$ is the number of \createF calls.
\end{theorem}

\begin{proof}
  The fast disjoint-sets data structure provides the bound of
  amortized time $O(\iack(m,n))$ per operation where $m$ is the number
  of operations and $n$ is the number of sets.  For our program, $m$
  is at most the number of memory accesses and $n$ is the number of
  \createF calls.  Note here, again, that unlike series parallel
  computations, each write may generate multiple queries; however, for
  the reason as explained in \secref{race}, the total number of
  queries is bounded by the two times the total number of memory
  accesses since each writer removes the entire reader-list.
  Therefore, the total running time is $O(T_1 \iack(m, n))$.
\end{proof}

\subheading{Intuition for the proof}

In order to argue that \MultiBags is correct, we must prove the
following theorem.

\begin{theorem}
\thmlabel{correctness}
If the currently executing strand is $v$, then a previously executed
strand $u$ is currently in an $S$ bag iff $u\prec v$.
\end{theorem}

In order to prove this theorem, we define two more terms. A node $u$
is a \defn{spawn predecessor} of a node $v$ if there is a path from
$u$ to $v$ which consists of only spawn and continue edges.  A node
$u$ is a \defn{join predecessor} of $v$ if there is a path from $u$ to
$v$ that consists of only join and continue edges.  Spawn and join
successors are defined in the symmetric way.  We will overload
notation and say that a strand $u$ is a spawn predecessor of a
function $F$ if there is a path from $u$ to the first strand of $F$
that consists of only spawn and continue edges and similarly a strand
$v$ is a join successor of $F$ if there is a path from the last strand
of $F$ to $v$.  Each node is its own spawn and join predecessor and
successor.

The algorithm works due to the following observations.\footnote{We
  implicitly assume that when we refer to any strand (or function), it
  is either currently executing or has already executed (we have no
  knowledge of strands or functions that are still to execute).}  We
  say that a function is \defn{active} if it has started executing,
  but has not completed (returned).  While a function $F$ is active,
  $S_F$ exists and $P_F$ doesn't because $P_F$ is only created upon
  $F$'s return.  All strands of an active function $F$ are in $S_F$.
  After a function has returned, $S_F$ is destroyed; $P_F$ exists
  if \getF has not been called on $F$'s future handle and if $P_F$
  exists, then all strands of $F$ are in $P_F$.  After a function has
  joined ($\getF$ has been called) then neither $S_F$ nor $P_F$ exist.

Property \ref{prop:eager} is a property of eager executions. 

\begin{property}
\label{prop:eager}
When a strand $v$ is currently executing, all spawn predecessors of
$v$ are part of some active function.  The converse is also true; all
strands $w$ that are part of active functions are spawn predecessors
of $v$.
\end{property}

We need two other properties. The first one is a static property of
paths in a program with structured futures as follows. If there is a
path from a strand $u$ to stand $v$,\footnote{In general, there can be
many paths from $u$ to $v$.} then there must be a path where the first
(possibly empty) part of the path consists of only join and continue
edges, while the second (possibly empty) part of the path contains
only spawn and continue edges. In other words, there is a path where
no join edges follow spawn edges. More formally, for any two nodes
$u$ and $v$, if $u\prec v$, then we can find a node $w$ where $u$ is a
join-predecessor of $w$ and $w$ is a spawn predecessor of $v$.
Combining with Property~\ref{prop:eager}, we get:
\begin{observation}
Consider a completed strand $u$ and a currently-executing strand $v$
where $u \prec v$. The furthest join successor of $u$, say $w$, must
be part of an active function.
\label{obs:predecessor}
\end{observation}
This observation allows us to concern ourselves only with paths that
go through nodes of active functions.  In particular, to detect races,
it is sufficient to try to check if the furthest join successor $w$ of
any previously executed node $u$ is part of an active function.  If
so, the observation implies that $u \prec w$ and we already know from
Property~\ref{prop:eager} that $w\prec v$; therefore, $u \prec v$.  If
not, then $u \not \prec v$.

The second property is a dynamic property of \MultiBags which allows
to precisely check this.  In particular, it states the following
(which combined with the previous observation gives us the theorem):
\begin{observation}
  Consider an already completed node $u$. Say, at time $t$, $u$'s
  furthest join successor $w$ is part of a function $G$.  If $G$ is
  active, then $u$ is in $G$'s $S$ bag, otherwise $u$ is in $G$'s $P$
  bag.
\label{obs:dynamic}
\end{observation}
An informal argument about why this property is true follows: In
\Liref{union}, \MultiBags unions $P_G$ into $S_F$ when $G$ joins with
an active function $F$.  This suggests the following: Consider a
particular strand $u$ in function $G$.  Say at time $t$, $w$ is the
furthest join successor strand of $u$ which has executed (or is
executing).  (This strand is well defined since there is no branching
in a path that contains only join and continue edges.)  Say $w$ is
part of function $F$.  Then at time $t$, $u$ is in $F$'s bag.  $F$'s
bag is an $S$ bag if $F$ is active.  Therefore, to check if $u$'s
furthest join successor $w$ is active, it suffices to check if $u$ is
in an $S$ bag, which is precisely what \MultiBags does.  

If we combine Observations~\ref{obs:predecessor}
and~\ref{obs:dynamic}, we get the theorem.  Say $v$ is executing at
time $t$ and consider a previously executed node $u$ and say $w$ is
the furthest join successor of $u$.  If $u \prec v$, then by the first
observation, $w$ is part of an active function and therefore, by the
second observation, $u$ is in an $S$ bag.  On the other hand, if
$u \parallel v$, then $w$ can not be part of an active function --- if
it was, then by Property~\ref{prop:eager}, $w \prec v$ and therefore,
$u \prec v$ (a contradiction).  Therefore, $u$ is in a $P$ bag.

When the computation does a memory access, it simply checks if the
previous accesses are in an $S$
In order to do queries, the algorithm operates as follows:
\begin{enumerate}
\item Currently executing strand $s$ reads location $\ell$:
  check if $\lastwriter(\ell)$ is in a $P$ bag; if so, declare a race.
  Otherwise, append $r$ to $\rlist(\ell)$
\item Currently executing strand $s$ writes to location $\ell$:
  check if any reader $r \in \rlist(\ell)$ is in a $P$ bag; if
  so, declare a race.  Otherwise, empty the reader list and set
  $\lastwriter(\ell)=w$.
\end{enumerate}

\subheading{Proof of Correctness of \MultiBags}

First, we prove the static property stated in
Observation~\ref{obs:predecessor}.  In order to prove it, we first
define a canonical order on the futures of the program and show that
we can always order the functions in this canonical order.

\begin{lemma}
\label{lem:canonical}
We can always find a canonical order.
\end{lemma}
\begin{proof}
  The order starts with the main function; we then progressively add
  futures to the computation.  A future $F$ can only be added if its
  creator strand and getter strand have already both been added.  

  We will induct on adding new futures.  We can always start since the
  main function is added first.  At some point, we have already added
  some futures, say a set $S$.  Next, we will add the future $F$ such
  that there is no future $G$ such that
  $\creator(F) \prec \creator(G)$ where both creator nodes are in $S$.
  There must be some such future (if there are more than one, we can
  choose arbitrarily).  Since no future was created sequentially after
  $F$ was created, and the strand before the $\getter(F)$ is
  sequentially after $\creator(F)$, $\getter(F)$ must also be in $S$;
  therefore, it is legal for us to pick $F$ as the next future in our
  canonical order.
\end{proof}

We now use this canonical order to induct on the futures.  In
particular, we can show that the static property stated in
Observation~\ref{obs:predecessor} (stated more formally and completely
in the following lemma) by inducting on the futures in the canonical
order.

\begin{lemma}\lemlabel{path}
  If $u\prec v$, then there exists a path from $u$ to $v$ that
  contains two sections: the first path (possibly empty) contains only
  join and continue edges and the second part (possibly empty)
  contains only spawn and continue edges.  In other words there is
  never a spawn edge followed by a join edge on this path.  In
  addition, this path is unique.  Therefore, If $u\prec v$, then there
  is some node $w$ (possibly $u$ or $v$) which is a join successor of
  $u$ and a spawn predecessor of $v$.
\end{lemma}

\begin{proof}
Induct on futures in the canonical order (which we can always find
  according to Lemma~\ref{lem:canonical}) and show that this is true
  as we add futures one by one.

Base case: We first have only the main strand, so this is true
  trivially.

Inductive case: Assume that after we have added a set $S$ of futures,
the statement is true.  We now add a new future $F'$.

Consider any nodes $u$ and $v$ in this new dag where $u \prec v$.  If
neither $u$ nor $v$ are in $F'$, then the addition of $F'$ does not
add any new paths between $u$ and $v$ (since the only new path added
is between $\creator(F)$ and $\getter(F)$ and there was already a path
between them before we added $F'$).  In addition, any new path added
does have a spawn followed by a join --- therefore, the uniqueness is
preserved.  Therefore, we only need consider pairs where either $u$ or
$v$ are in $F'$.  If $u$ is in $F'$ and $v$ is not, then the path from
$u$ to $v$ must go from the last strand of $F'$ to $\getter(F')$ and
then to $v$.  By inductive hypothesis, the path from $\getter(F')$
already follows the desired property and the path from $u$ to
$\getter(F')$ only contains join and continue edges.  Therefore, the
property still holds.  A symmetric argument applies when $v$ is in
$F'$.
\end{proof}

We then prove the dynamic property stated in
Observation~\ref{obs:dynamic} by looking at the execution as it
unfolds. For each function $F$, we define its \defn{operating
function} $G$ as the function containing the ``furthest join
descendant'' of the last executed strand of $F$.  An active function
is its own operating function. If a function is not active (it has
returned), it may be \defn{confluent} or \defn{non-confluent}.  It is
confluent if its operating function is active; otherwise it is
non-confluent.  By definition, the operating function of a
non-confluent function is always non-confluent.  A confluent function
can never be its own operating function, but a non-confluent function
$F$ may be its own operating function if its $\getter(F)$ has not yet
executed.


The following lemma is proved by induction on the program as it
executes.

\begin{lemma}\lemlabel{SPStuff}
(a) When a function $F$ is active, all its strands are in its $S$
bag. (b) If a function $F$ is confluent, then all its strands are in
its operating function $G$'s $S$ bag.  (c) If a function $F$ is
non-confluent, then all its strands are in its operating function
$G$'s $P$ bag.
\end{lemma}

 \begin{proof}
When a function is first called, it has an $S$ bag and its strands are
placed in the $S$ bag.  They remain in this $S$ bag while it is
active.  Once the function returns, all its items move to a $P$ bag.
For the other two statements, we induct on time after $F$ returns.

Base Case: When $F$ returns, $\getter(F)$ has not yet been called.
Therefore, it is its own operating function; it is non-confluent; and
all its strands are in its own $P$ bag.

Inductive Case: We will do this by two cases: 

\textbf{Case 1:} $F$ is non-confluent and $G$ is its (non-confluent)
operating function; by inductive hypothesis, all strands of $F$ are
in $G$'s $P$ bag.  The only thing that can make changes the location
of its strands is if $\getter(G)$ executes, say by function $H$.  At
this point, $H$ (which is currently active) becomes the operating
function for both $G$ and $F$ --- therefore, $F$ is now confluent.
All strands of $F$ (and incidentally $G$) move to $H$'s $S$ bag.
   
\textbf{Case 2:} $F$ is confluent and $G$ is its (active) operating
function; by inductive hypothesis, all strands of $F$ are in $G$'s
$S$ bag.  The only thing that changes the location of $F$'s strands is
if $G$ returns.  At this point $G$ becomes non-confluent (since it is
no longer active); therefore $F$ also becomes non-confluent.  All of
$F$'s strands move to $G$'s $P$ bag.
\end{proof}

The combination of static and dynamic properties leads to the proof of
correctness. The intuition is that if a function $F$ is confluent,
then there is some strand $w$ in its (active) operating procedure
which is a join successor of all strands of $F$ and a spawn
predecessor of currently executing strand.

\begin{reptheorem}{thm:correctness}
If the currently executing strand is $v$, then a previously executed
strand $u$ is currently in an $S$ bag iff $u\prec v$.
\end{reptheorem}

\begin{proof}
  By Lemma~\ref{lem:path}, we know that if $u\prec v$, then we can
  find a node $w$ such that $u$ is a join predecessor of $w$ and $w$
  is a spawn predecessor of $v$.  By Property~\ref{prop:eager}, since
  $v$ is executing, the function containing $w$, say $G$, is still
  active.  Therefore, by definition, the function containing $u$ is
  confluent.  Therefore, by Lemma~\ref{lem:SPStuff}, $u$ is an $S$
  bag.

  If $u$ does not precede $v$, then there is no path from $u$ to $v$.
  Therefore $u$ can not have a path to any strand $w$ in any active
  function (otherwise by the second statement of
  Property~\ref{prop:eager}, since $w$ has path to $v$, $u$ will also
  have a path to $v$).  Therefore, by definition, $u$ is non
  confluent.  By Lemma~\ref{lem:SPStuff}, $u$ is in a $P$ bag.
\end{proof}

\secput{general}{\MultiBagsPlus for General Futures} 

We now consider general use of futures for programs that use both \spawn/\sync
constructs and also futures.  In particular, we consider programs where most
of the parallelism is created using \spawn and \sync, but there are also $k$
future \getF operations.  For these programs, we provide a race detection
algorithm that runs in total time $O(T_1\iack(m,n)+k^2)$, where $T_1$ is the
work of the program, $\iack$ is the inverse Ackermann's function, $m$ is the
number of memory accesses in the program and $n$ is the number of \spawn and
\createF calls.  To put this bound in context, a series-parallel program has
$k=0$ --- in this case (and in fact, for any program where $k =
O(\sqrt{T_1})$), the \MultiBagsPlus runs in time $O(T_1\iack(m,n))$.  Since
the inverse Ackermann's function grows slowly (upper bounded by $4$), this
bound is close to asymptotically optimal.  

As mentioned in \secref{prelim}, \MultiBagsPlus depends on 
eager execution of the computation and  
we assume that our futures are forward-pointing.  Therefore, the depth-first 
execution never blocks on a \getF call since the corresponding future has already
finished executing.


\paragraph{Notation:}
Unlike in \secref{structured}, we must distinguish between \spawn 
and \createF (similarly, between \sync and \getF) for \MultiBagsPlus.  
The computation dag consists of five kinds of nodes: (1)
regular strands with one incoming and one outgoing edge; (2)
\defn{spawn} strands which end with a \spawn instruction and have with
two outgoing edges; (3) \defn{creator} strands which end with a
\createF instruction and have with two outgoing edges; (4) \defn{sync}
strands which begin immediately after a instruction and have with two
incoming edges; and (5) \defn{getter} strands which begin immediately
after \getF instruction and have with two incoming edges.  Some
strands can have two incoming and two outgoing edges (if they start
immediately after a \getF or \sync instruction and end with a \spawn 
or \createF); these strands are correspondingly in both categories.

The computation dag also consists of five kinds of edges: \defn{spawn}
edges are from spawn nodes to the first strand of the
corresponding spawned function; \defn{join} edges are from last
strand of a spawned function to the corresponding sync node;
\defn{create} edges from the creator strand to the first stand of the
future function; and \defn{get} edges from the last strand of a future
function to the corresponding getter node. Each future can have
multiple get edges if it is a multi-touch future.\footnote{For
  context, in Section~\ref{sec:structured}, both spawn and create
  edges were called spawn edges and both join and get edges were
  called join edges.}

\paragraph{Reachability data structures:} Recall that we can model
computations that employ futures as a set of series-parallel dags (SP dags)
plus some non-SP edges (\secref{prelim}).  When we need to check if $u \prec
v$, if they are already in the same SP dag (i.e., $\spDag(u) = \spDag(v)$, as
defined in \secref{prelim}), the disjoint-sets data structure maintained by
\MultiBags can readily answer the reachability query correctly.  We only run
into trouble due to use of general futures when $\spDag(u) \neq \spDag(u)$ but
a path exists between them via (possibly more than one) non-SP edges.

Thus, \MultiBagsPlus maintains two data structures, and both utilize fast
disjoint-sets data structure described in \secref{structured}.  The first
disjoint-sets data structure, called $\dssp$, is virtually
identical to the data structure used in \secref{structured}: we maintain
$S$ and $P$ bags for each function, and the operations performed on
$\dssp$ are identical to those described in \secref{structured} in
the case of \createF, and \spawn is treated in the same way as \createF.  In
addition, \sync is treated like \getF.  The only difference is that we do not
perform any operation on this data structure on \getF (since we allow
multi-touch futures).  We can use $\dssp$ to correctly answer reachability
query between two strands if they are in the same SP dag.  Intuitively, for 
reasons similar to the structured case, all strands that are currently stored 
in an $S$ bag are sequentially before the currently executing strand.  
Note that all strands that are stored in a $P$ bag are not necessarily in 
parallel with the current strand, due to non-SP edges --- we will use the 
second data structure to answer that query.

The second data structure handles the additional complication in reachability
query when two strands are connected via non-SP edges.  This data structure
has two components: (1) a disjoint-sets data structure called $\dsnsp$ that
maintains a collection of disjoint-sets, and each strand is added to $\dsnsp$
when encountered; and (2) a separate dag called $\cal R$ that contains some of
the sets from $\dsnsp$.  The high-level idea is that these sets are made of
connected series-parallel subdags of the original dag $\gfull$.  For any two
nodes $u$ and $v$ in different SP dags, \MultiBagsPlus ensures that $u \prec
v$ in $\gfull$ iff $\Find(\dsnsp,u) \prec \Find(\dsnsp, v)$ in $\cal R$ (the
sets they are in are connected in $\cal R$).

We call sets also in $\cal R$ as the \defn{attached sets}, which store nodes that
are subdags which start and/or end with $\creator$ or $\getter$
strands.\footnote{This is not quite accurate; for technical reasons, some
attached sets start/end with regular, spawn, and join nodes as well.} 
$\cal R$ explicitly maintains reachability relationship that arises due to
non-SP edges between nodes in the attached sets.  $\cal R$ is simply a dag
(with each node being an attached set), but it is not series-parallel.  Thus,
to answer reachability queries quickly between nodes in $\cal R$,
\MultiBagsPlus maintains a full transitive closure of all sets in $\cal R$ ---
whenever a set is added to $\cal R$, its reachability from all sets already
added to $\cal R$ is explicitly computed and stored.  Therefore, one can check
if $A \prec B$ in $\cal R$ in constant time.

If every set could be in $\cal R$ we would be done.  We must keep $\cal R$
small, however, since every time we add a set to $\cal R$ we compute a full
transitive closure, which is expensive.  It turns out that it is difficult to
simultaneously put all strands in attached sets and keep $\cal R$ small.  In
order to cope with this, some strands are in \defn{unattached sets}, which are
only stored in $\dsnsp$.  Intuitively, an unattached set contains nodes of a
complete series-parallel subdag which have no incident non-SP edges.  Each
unattached set $U$ has two additional fields, \defn{attached predecessor} and
\defn{attached successor}, which point to attached sets that act as $U$'s
proxies when querying $\cal R$.  $U$'s attached predecessor, denoted as
$U.attPred$, is set when $U$ is created; therefore, it always points to some
attached set.  $U$'s attached successor, denoted as $U.attSucc$, is set at
some later point; it either points to some attached set or may be null.  An
attached set is always its own attached predecessor and successor.  We will
overload notation and say that node $u$'s attached predecessor is
$\find{\dsnsp, u}$'s attached predecessor (and similarly for attached
successor).

\begin{figure*}
\footnotesize
\begin{tabular}{|l|l|}
\hline
\begin{minipage}[t]{0.48\linewidth}
\vspace{-3mm}
\begin{codebox}
  \Procname{$\proc{Query}(u,v)$ \Comment return \const{true} iff $u\prec v$ in $\gfull$}
    \li \If $\find{\dssp, u}$ is an $S$-bag, \lilabel{part1Start}
        \RComment{Query $\dssp$ first}
    \li \Then \Return \const{true} \lilabel{part1End}
        \End
    \li $S_v = \find{\dsnsp, v}$
    \li \If $S_v$ is unattached
    \li \Then $S_v \gets \attpred{S_v}$
        \End
\end{codebox}
\end{minipage}%
&
\begin{minipage}[t]{0.51\linewidth}
\vspace{-3mm}
\begin{codebox*}
    \li $S_u = \find{\dsnsp,u}$ \lilabel{part2Start}
    \li \If $S_u$ is unattached
    \li \Then $S_u \gets \attsucc{S_u}$
    \li    \If $S_u = \const{null}$ \Return \const{false} \End 
        \End
    \li $ans =$ query $\cal R$ to determine if $S_u \prec S_v$ \lilabel{part2End}
    \li \Return $ans$
\end{codebox*}
\end{minipage}%
\\ \hline
\end{tabular}
\caption{Code for querying reachability.}
\label{code:query}
\vspace{-2mm}
\end{figure*}

\begin{figure*}
\footnotesize
  \begin{tabular}{|l|l|}
  \hline
  \begin{minipage}[t]{0.48\linewidth}
  $u$ is the first strand of the computation:\vspace{-1ex}
  \begin{codebox}
    \li add $u$ to an attached set with no predecessor.  
  \end{codebox}%
  \vspace{-2ex}
\noindent\rule{\linewidth}{0.4pt}
  Function $F$ calls \spawn($G$):\vspace{-1ex}
  \begin{codebox*}
    \zi \RComment $u$ is the strand in $F$ immediately before the \spawn
    \zi \RComment $v$ is the strand in $F$ right after \spawn 
    \zi \RComment $w$ is the first strand of $G$
    \li $S_G \gets \MakeSet(\dssp, w)$ \lilabel{dsspspawn}
    \li $U_v \gets \MakeSet(\dsnsp, v)$ and make $U_v$ unattached.
    \li $\attpred{U_v} \gets \attpred{\find{\dsnsp, u}}$ \lilabel{setAttPredSpawn}
    \li $U_w \gets \MakeSet(\dsnsp, w)$ and make $U_w$ unattached.
    \li $\attpred{U_w} \gets \attpred{\find{\dsnsp, u}}$ \lilabel{setAttPredCont}
  \end{codebox*}
  \vspace{-.5ex}
\noindent\rule{\linewidth}{0.4pt}
  Function $F$ calls $\createF(G)$:\vspace{-1ex} 
  \begin{codebox*}
     \zi \RComment $u$ is the strand in $F$ immediately before the \createF
     \zi \RComment $v$ is the strand in $F$ immediately after the \createF 
     \zi \RComment $w$ is the first strand of $G$
     \li $S_G \gets \MakeSet(\dssp, w)$ \lilabel{dsspcreate}
     \li $\proc{Attachify}(u)$   \lilabel{AttachifyCreate1}
     \li $A_v \gets \MakeSet(\dsnsp, v)$ and make $A_v$ attached. \lilabel{AttachifyCreate2}
     \li Add an arc from $\find{\dsnsp, u}$ to $A_v$ in $\cal R$. \lilabel{arcCreatorFut}
     \li $A_w \gets \MakeSet(\dsnsp, w)$ and make $A_w$ attached.  \lilabel{AttachifyCreate3}
     \li Add an arc from $\find{\dsnsp, u}$ to $A_w$ in $\cal R$ \lilabel{arcCreatorCont}
  \end{codebox*}
  \vspace{-.5ex}
\noindent\rule{\linewidth}{0.4pt}
  Function $G$ returns:\vspace{-1ex}
  \begin{codebox*}
     \li $P_G \gets S_G$; deallocate $S_G$ \lilabel{dsspreturn}
  \end{codebox*}
  \vspace{-2ex}
\noindent\rule{\linewidth}{0.4pt}
  Function $F$ calls \getF($G$):\vspace{-1ex}
  \begin{codebox*}
     \zi \RComment $u$ is the strand in $F$ that ended with the get 
     \zi \RComment $v$ is the strand immediately after $u$ in $F$ 
     \zi \RComment $w$ is the last strand of $G$.  
     \li $\proc{Attachify}(u)$  \lilabel{AttachifyGet1}
     \li $A_v \gets \MakeSet(\dsnsp, v)$ and make $A_v$ attached. \lilabel{AttachifyGet2}
     \li Add an arc from $\find{\dsnsp, u}$ to $A_v$ in $\cal R$  \lilabel{arcContGetter}
     \li Add an arc from $\find{\dsnsp, w}$ to $A_v$ in $\cal R$; 
     \zi \RComment $\find{\dsnsp, w}$ is guaranteed to be attached. 
  \end{codebox*}
  \end{minipage}%
  & 
  \begin{minipage}[t]{0.51\linewidth}
  \vspace{-2ex}
  \begin{codebox*}
    \Procname{$\proc{Attachify}(u)$ 
        \Comment{make the set containing $u$ attached if not already.}}
    \li $U_u \gets \proc{Find}(\dsnsp, u)$
    \li \If $U_u$ is unattached
    \li \Then mark $U_u$ as attached
    \li       add $U_u$ to $\cal R$
    \li       add the arc $(\attpred{U_u},U_u)$ to $\cal R$
        \End
  \end{codebox*}
  \vspace{-.5ex}
\noindent\rule{\linewidth}{0.4pt}
  Function $F$ calls \sync with child function $G$:\vspace{-1ex}
  \begin{codebox*}
    \li $S_F \gets \proc{Union}(\dssp, S_F,P_G)$; deallocate $P_G$ \lilabel{syncBegin}
    \li look at the corresponding fork
    \li let $f$ be the strand immediately preceding the fork
    \li let $s_1$ and $s_2$ be $f$'s two immediate successors of $f$
    \zi \RComment i.e., the first strand of $G$ and the first strand of 
                  the continuation
    \li let $j$ be the strand immediately after the sync
    \li let $t_1$ and $t_2$ be $j$'s immediate predecessors of the sync
    \zi \RComment i.e., the last strand of the $G$ and the continuation
    \li \If neither $\find{\dsnsp, t_1}$ nor $\find{\dsnsp, t_2}$ is attached
    \zi \RComment{No non-SP edges} \lilabel{unattachedUnionStart}
    \li \Then $\proc{Union}(\dsnsp, f,t_1)$ 
    \li    $\proc{Union}(\dsnsp, f,t_2)$
    \li    $\proc{Union}(\dsnsp, f,\MakeSet(j))$ \lilabel{unattachedUnionEnd}
    \li \ElseIf both $\find{\dsnsp, t_1}$ and $\find{\dsnsp, t_2}$ are attached \lilabel{twoAttachedSync}
    \li \Then $\proc{Attachify}(f)$ \lilabel{spawnAttachify}
    \li    add arc $(\find{\dsnsp, f},\find{\dsnsp, s_1})$ to $\cal R$ \lilabel{arcSourceChildren1}
    \li    add arc $(\find{\dsnsp, f},\find{\dsnsp, s_2})$ to $\cal R$ \lilabel{arcSourceChildren2}
    \li    $A_j = \MakeSet(\dsnsp, j)$ and make $A_j$ attached \lilabel{joinAttachify}
    \li    add a node $A_j$ to $\cal R$
    \li    add arc $(\find{\dsnsp, t_1},A_j)$ to $\cal R$ \lilabel{arcChildrenSink1}
    \li    add arc $(\find{\dsnsp, t_2},A_j)$ to $\cal R$ \lilabel{arcChildrenSink2}
    \li \Else let $t_a$ be the attached one and $t_u$ be the \lilabel{oneAttachedSync}
                unattached one
    \li   correspondingly $s_a$ is attached and $s_u$ is unattached
    \li   \If $\find{\dsnsp, f}$ is not attached 
    \li   \Then $\proc{Union}(\dsnsp, s_a,f)$ \lilabel{singleUnionSrc}
          \End
    \li   $\proc{Union}(\dsnsp, t_a, \proc{Make-Set}(j))$ \lilabel{singleUnionSync}
    \li         $\attsucc{\find{\dsnsp, t_u}} \gets \find{\dsnsp, j}$ \lilabel{setAttSucc}
        \End
  \end{codebox*}
  \end{minipage}%
  \\ \hline
  \end{tabular}
\caption{The actions taken by the algorithm to maintain $\dssp$, $\dsnsp$ and $\cal R$}
\label{code:maintainReach}
\vspace{-2mm}
\end{figure*}

\paragraph{Answering queries:} Figure~\ref{code:query} shows how the
reachability data structures are queried to find out if a
path exists between some previously executed node $u$ and the currently
executing node $v$.  In the first part of the query
(\lirefs{part1Start}{part1End}), we query $\dssp$ and if $u$ is in the
$S$ bag, then we can conclude that $u \prec_\full v$ and return.  If
$u$ is in the $P$ bag, then we check if attached successor of $u$
precedes the attached predecessor of $v$ in $\cal R$; if so, we say
that $u \prec_\full v$.  Otherwise $u$ is in parallel with $v$.

\begin{figure}
\begin{center}
\includegraphics[width=3.0in]{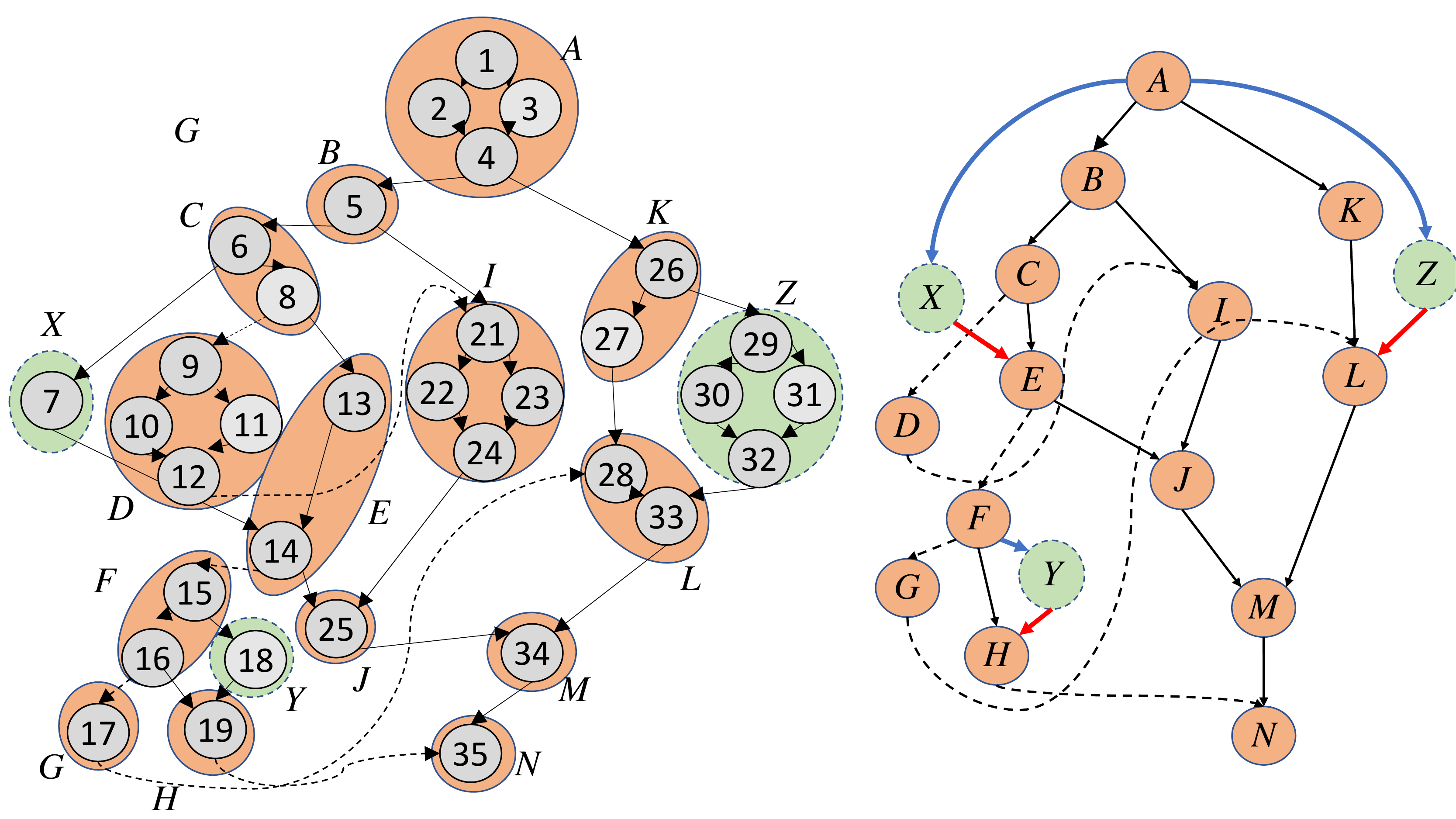}
\end{center}
\vspace{-3mm}
  \caption{An example of general futures.  The left figure shows the full dag.  
    The dashed edges are \createF and \getF edges.  The numbering shows 
    the order in which the nodes execute in depth-first eager execution.  
    The orange ovals (solid outline) are attached sets while the green 
    ovals (dashed outlines) are unattached.  The right figure shows 
    $\cal R$ at the end of the computation.  The thick blue lines 
    indicate an attached predecessor relationship while a thick red 
    line indicates an attached successor relationship.}
    \figlabel{fullExample}
    \vspace{-3mm}
\end{figure}

\paragraph{Maintaining $\dssp$, $\dsnsp$, and $\cal R$:}
Figure~\ref{code:maintainReach} shows the code for maintaining reachability
relationships between nodes in the computation.  
The first thing we do during a \spawn, \createF, \texttt{return} and \sync is
to manipulate $\dssp$
(\lireffour{dsspspawn}{dsspcreate}{dsspreturn}{syncBegin}) in a manner
identical to \secref{structured}.\footnote{Since we are assuming binary
forking, we sync with one function at a time.}  

Now lets consider the manipulations of $\dsnsp$ and $\cal R$. It uses an
auxiliary function $\proc{Attachify}(u)$, which simply checks if $U_u =
\find{\dsnsp, u}$ is an unattached set, and if so, converts it into an
attached set by adding it to $\cal R$ and adding an edge from $\attpred{U_u}$
to $U_u$ in $\cal R$.

The attached and unattached sets change as the execution continues.
\MultiBagsPlus unions sets in $\dsnsp$ growing both attached and
unattached sets.  Two attached sets are never unioned together.
Whenever we union an attached set and an unattached set, we always
union the unattached set into the attached set; therefore, the
resulting set is attached and remains in $\cal R$.  On the other hand,
an unattached set contains nodes of a complete series-parallel subdag
which have no incident non-SP edges.  In particular, consider a
parallel composition of two series-parallel subdags $G_1$ and $G_2$.
Say $G_1$ has no incident non-SP edges.  Then all nodes of $G_1$
constitute an unattached set if either the join node that joins $G_1$
and $G_2$ has not executed yet, or if $G_2$ has an incident non-SP
edge.  At a high-level, this design allows us to have the property
that each non-SP edge only leads to a constant number of attached
sets.  Since only attached sets are in $\cal R$, this idea allows us
to keep $\cal R$ small, allowing us to get good performance.  

Figure~\ref{fig:fullExample} shows a computation dag with futures and
its corresponding $\cal R$ at the end of the computation.  As can be
seen in this example, sets $X, Y$ and $Z$ are all unattached and each
have the above property.  In Appendix~\ref{sec:generalOmitted},
Figure~\ref{fig:partialExample} shows the dag when it has been
partially executed while nodes 23 and 31 are executing, respectively
and all unattached sets in those examples also have this property.

Due to space limitations, all the proofs are omitted to
Appendix~\ref{sec:generalOmitted}.  Here we simply state the performance and
correctness theorems.  

\newcommand{\performancethmgeneral}{
\MultiBagsPlus detects races in time $O(T_1\iack(m,n) + k^2)$ for
programs with $T_1$ work, $k$ \getF calls, $m$ memory accesses, and 
$n$ number of strands.  
}

\begin{theorem}
\thmlabel{generalPerfBound}
\performancethmgeneral
\end{theorem}


\newcommand{\correctnessthmgeneral}{ If the program is executed in a
  depth first eager execution order, $\proc{Query}(u,v)$ returns true
  iff $u\prec v$ in $\gfull$.}

\begin{theorem}
\label{thm:unstructCorrect}
\correctnessthmgeneral
\end{theorem}


\paragraph{Comparison to algorithm by Agrawal et al.\@~\cite{AgrawalDeFi18}:}
As we have fully described the \MultiBagsPlus algorithm, we discuss the
differences between \MultiBagsPlus and the the-of-the-art algorithm by
Agrawal et al.\@~\cite{AgrawalDeFi18} and provide an analytical analysis
as to why the algorithm by Agrawal et al.\@ is much more challenging to 
implement in practice.

The algorithm by Agrawal et al.\@ utilizes the following data structures to
answer reachability queries: 1) an order-maintenance data structure for
answering series-parallel queries; 2) the full computation DAG to update and
maintain ``anchor-predecessors'' and ``proxies'' used to infer
``anchor-successors;'' and 3) a reachability matrix $\cal R$ which contains
anchor nodes to answer reachability queries involving non-SP edges.  The
functionalities served by these data structures are similar to that of
$\dssp$, $\dsnsp$, and $\cal R$ in \MultiBagsPlus; in particular, their
algorithm utilize anchor-predecessors and anchor-successors to allow for
correct reachability queries involving non-SP edges, sharing similar roles as
the attached predecessors and attached successors in \MultiBagsPlus.  The main
difference is in the second data structure and how the anchor-predecessors and
anchor-successors are maintained.

In the algorithm by Agrawal et al., the mechanism for maintaining
anchor-predecessors and proxies (which are used to infer anchor-successors)
are more complex.  In particular, to maintain anchor-predecessors, the
algorithm maintains the full computation dag, and each strand (a node in the
dag) explicitly stores its anchor-predecessor.  However, anchor-predecessors
can sometimes change as the program executes.  When that occurs, the algorithm
must explicitly traverse subpart of the dag and update some of the
predecessors explicitly.  The asymptotic complexity of such updates is still
ok because the paper argues that a strand's anchor predecessor can only change
a constant number of times.

Similarly, the algorithm maintains a proxy per strand, used to infer a
strand's anchor-successor.  A proxy for a strand is stored instead of its
anchor-successor is because, while an anchor-predecessor of a node can change
a constant number of times, its anchor-successor can change many times.  Thus
instead, the algorithm maintains a proxy, which indirectly allows the
algorithm to deduce its anchor-successor.  Like the anchor-predecessor, a
proxy of a node can only change a constant number of times, and when that
occurs, the algorithm again explicitly traverses the relevant subdag and
updates the proxies explicitly.

We argue that this algorithm is harder to implement and likely has higher
overheads due to the following reasons.  First, explicitly maintaining the
entire program dag and also storing each strand's anchor-predecessor and proxy
would be more memory intensive than keeping these strands in union-find data
structures which are tagged appropriately.  Second, explicit dag traversals in
order to update proxies and anchor-predecessors of nodes would be expensive
(even though the asymptotic complexity is manageable).  This prior work
establishes the state-of-the art time bound for race detecting programs that
use general futures, but no implementation exists.

\secput{exp}{Experimental Evaluation}

This section empirically evaluates \FutureRD that implements 
\MultiBags and \MultiBagsPlus described earlier.  We first evaluate 
the practical efficiency of these algorithms and then the performance 
difference between them, focusing on the impact of the additional $k^2$ 
overhead that \MultiBagsPlus incurs, where $k$ is the number of \getF 
operations.  

\subsection*{Implementation of \FutureRD}
 
\FutureRD works by instrumenting parallel program executions: upon
the execution of a parallel construct
(i.e., \spawn, \sync, \createF, and \getF), it invokes the necessary
operations to update the reachability data structures; likewise, upon
the execution of a memory access, it invokes the necessary operations
to update the access history data structure and query both data 
structures.

We use Intel Cilk Plus~\cite{IntelCilkPlus13} as our language front end, which
is a C/C++ based task parallel platform that readily supports fork-join
parallelism.  Cilk Plus does not currently support the use of futures,
however, so we have implemented our own future library. Since our race
detector executes the program sequentially with eager evaluation of futures,
the future library never actually interacts with Cilk Plus runtime during race
detection.

Both \MultiBags and \MultiBagsPlus utilize disjoint-sets data structures to
maintain reachability as described in \secref{structured}). \MultiBagsPlus
additionally needs to maintain $\cal R$ as part of its reachability data
structure (defined in \secref{general}). Conceptually, $\cal R$ is simply a
boolean reachability matrix where each cell $(i, j)$ indicates whether
there is a path from attached set $i$ to attached set $j$.  \FutureRD
maintains $\cal R$ as a vector of bit vectors, representing the reachability
between any two sets using a single bit.  Whenever an edge is added to $\cal
R$, reachability is transitively propagated via parallel bit operations. 


\FutureRD maintains the access history like a two-level direct-mapped cache, 
and keeps track of the reader list and last writer at four-byte
granularity (all our benchmarks perform four-byte or larger accesses).
That is, to query or update readers/writers for an address $a$, the more
significant bits of $a$ are used to index into the top-level table and the
rest of the bits are used to index into the second-level table.

\subsection*{Experimental setup}

We evaluate \FutureRD using six benchmarks:
longest-common subsequence (\texttt{lcs}), Smith-Waterman (\texttt{sw}),
matrix multiplication without temporary matrices (\texttt{mm}), binary tree
merge (\texttt{bst}) as described by \citet{BlellochRe97}, Heart Wall Tracking
(\texttt{heartwall}), and Dedup (\texttt{dedup}). Heart Wall Tracking and
Dedup both contain parallel patterns that cannot be easily implemented using
fork-join constructs alone. The Heart Wall Tracking algorithm is adapted from
the Rodinia benchmark suite~\cite{CheBoMe09} that tracks the movement of a
mouse heart over a sequence of ultrasound images.  
Dedup is a compression program that exhibits pipeline
parallelism~\cite{BieniaLi10}, taken from the Parsec benchmark
suite~\cite{BieniaKuSi08}.  All but \texttt{dedup} have two implementations:
structured and general futures; \texttt{dedup} does not utilize the flexibility
of general futures. 
We use the following input sizes: \texttt{lcs} uses $N = 16k$, \texttt{mm} and
\texttt{sw} use $N = 2048$, \texttt{heartwall} uses $10$ images,
\texttt{dedup} uses input large, and \texttt{bst} 
uses input tree sizes $8\mathrm{e}6$ and $4\mathrm{e}6$.
For \figreftwo{structured-overhead}{general-overhead}, we use base case $B =
\sqrt{N}$ for \texttt{lcs}, \texttt{mm}, and \texttt{sw} to keep the work 
the same for the baseline, \MultiBags, and \MultiBagsPlus (since \MultiBagsPlus has
$k^2$ additional overhead).  We then vary the base case size for
\figref{perf-diff}. 

We ran our experiments on an Intel Xeon E5-4620 with $32$ 2.20-GHz cores on
four sockets. Each core has a 32-KByte L1 data cache, 32-KByte L1 instruction
cache, a 256-KByte L2 cache.  There is a total of 500 GB of memory, and
each socket shares a 16-MByte L3-cache. All benchmarks are compiled with
LLVM/Clang 3.4.1 with \code{-O3 -flto} running on Linux kernel version 3.10.
Each data point is the average of $5$ runs with standard deviation less than
$5\%$ with the exception of running \texttt{dedup} with full race detection,
which sees a standard deviation under $9\%$.

\subsection*{Practical efficiency of \FutureRD}

\addtolength{\textfloatsep}{-5mm}

First, we evaluate the overhead of \FutureRD and show that the algorithms
can be implemented efficiently.  To get the sense of
where the overhead comes from, we ran the application benchmarks with
four configurations: 
\begin{closeitemize}
\item \defn{baseline}: running time without race detection; 
\item \defn{reachability}: running time with only the reachability components,
including the instrumentation overhead to capture parallel control constructs;
\item \defn{instrumentation}: running time with memory-access instrumentation 
overhead on top of the reachability configuration, but does not maintain 
or query the access history;
\item \defn{full}: running time with the full race detection overhead.
\end{closeitemize}

\begin{figure}[h]
\centering
\begin{small}
\setlength\tabcolsep{4pt}
\newcommand{\mhead}[2]{\multicolumn{#1}{c}{\it #2}}
\newcommand{\oh}[1]{\hspace{1pt}\scriptsize (#1$\times$)}
\begin{tabular}{crr@{\hspace{1pt}}rr@{\hspace{1pt}}rr@{\hspace{1pt}}r}
\mhead{1}{bench} & \mhead{1}{baseline} & \mhead{2}{reachability} 
       & \mhead{2}{instr} & \mhead{2}{full} \\
\hline
\texttt{lcs} &  2.19 
    &  2.23 & \oh{1.02} &  6.65 & \oh{3.04} &  54.27 & \oh{24.77} \\
\texttt{sw}  & 14.78 
    & 14.25 & \oh{0.96} & 28.79 & \oh{1.95} & 325.10 & \oh{22.00} \\
\texttt{mm}  & 13.94 
    & 13.82 & \oh{0.99} & 58.84 & \oh{4.22} & 468.75 & \oh{33.61} \\
\texttt{heartwall} & 13.86 
    & 13.77 & \oh{0.99} & 63.39 & \oh{4.58} & 340.04 & \oh{24.54} \\
\texttt{dedup} & 12.38 
    & 12.15 & \oh{0.98} & 13.79 & \oh{1.11} &  26.43 & \oh{2.14} \\
\texttt{bst} & 1.37 
    &  1.92 & \oh{1.41} &  2.65 & \oh{1.94} &  10.94 & \oh{8.02} \\
\end{tabular}
\vspace{-2mm}
\caption{The execution times for the benchmarks using structured futures,
    shown in seconds, with \MultiBags used for race detection.
    Numbers in the parentheses show the overhead compared to the baseline.}
\label{fig:structured-overhead}
\end{small}
\vspace{-3mm}
\end{figure}

\begin{figure}[h]
\centering
\begin{small}
\setlength\tabcolsep{4pt}
\newcommand{\mhead}[2]{\multicolumn{#1}{c}{\it #2}}
\newcommand{\oh}[1]{\hspace{1pt}\scriptsize (#1$\times$)}
\begin{tabular}{crr@{\hspace{1pt}}rr@{\hspace{1pt}}rr@{\hspace{1pt}}r}
\mhead{1}{bench} & \mhead{1}{baseline} & \mhead{2}{reachability} 
       & \mhead{2}{instr} & \mhead{2}{full} \\
\hline
\texttt{lcs} &  2.03 
    &  2.30 & \oh{1.14} &  6.47 & \oh{3.19} &  54.95 & \oh{27.13} \\
\texttt{sw}  & 14.73 
    & 14.65 & \oh{0.99} & 27.87 & \oh{1.89} & 380.19 & \oh{25.82} \\
\texttt{mm}  & 13.13 
    & 15.07 & \oh{1.15} & 64.04 & \oh{4.88} & 498.65 & \oh{37.99} \\
\texttt{heartwall} & 13.82 
    & 13.89 & \oh{1.00} & 56.58 & \oh{4.09} & 487.95 & \oh{35.31} \\
\texttt{dedup} & 12.11 
    & 27.73 & \oh{2.29} & 29.60 & \oh{2.44} &  52.39 & \oh{4.33} \\
\texttt{bst} & 1.44 
    &  6.01 & \oh{4.16} &  6.79 & \oh{4.70} &  18.18 & \oh{12.60} \\
\end{tabular}
\vspace{-2mm}
\caption{The execution times for the benchmarks using general futures,
    shown in seconds, with \MultiBagsPlus used for race detection.
    Numbers in the parentheses show the overhead compared to the baseline.}
\label{fig:general-overhead}
\end{small}
\end{figure}

\figref{structured-overhead} shows the list of programs that employ structured
futures running with different configurations, where \FutureRD maintains
reachability using the \MultiBags algorithm.  First, observe that the
reachability configuration incurs almost no overhead, except for \texttt{bst},
which has very little work per parallel construct.   
Since the operations on the disjoint-sets data structure are very
efficient, as long as there is sufficient work per parallel construct, the
overhead of maintaining reachability in \MultiBags should be low.  
These program contains large number of memory accesses, however, and thus
adding instrumentation for memory accesses alone incurs additional
$2$--$4.5\times$ overhead.

Going from the instrumentation configuration to the full race detection incurs
another $6$--$10\times$ overhead, with the exception of \texttt{dedup}.  We 
expect the additional overhead incurred to be about
$8$--$10\times$ because the full configuration transforms every
memory access into updates to access history and queries to both access 
history and reachability data structures.  Thus, each memory access is
translated into a few function calls and several pointer chases 
to multiple data structures.  The benchmark \texttt{heartwall} only incurs 
additional $6\times$, because it spends non-negligible amount of time 
performing I/O (reading in image files).  Finally, \texttt{dedup} is an 
outlier because \texttt{dedup} calls into a dynamic library to perform 
compression, which we could not recompile to include instrumentation.  
Thus, any memory accesses performed within the library do not incur 
additional overhead.  Since the compression takes up a
substantial amount of execution time, the additional overhead is small.

\figref{general-overhead} shows the runtime of programs that employ general
futures where \FutureRD maintains reachability using the \MultiBagsPlus
algorithm.  The additional overhead incurred going from one configuration to
the next is similar to \figref{structured-overhead} except the higher overhead
from \MultiBagsPlus is evident in the reachability configuration.

Over five benchmarks (excluding \texttt{dedup}, since we could not
instrument its compression library), we see a geometric mean
overhead of $1.06\times$ and $1.40\times$ to maintain reachability
using \MultiBags and \MultiBagsPlus, respectively. Full
race detection exhibits $20.48\times$ and $25.98\times$ overhead, respectively.

\subsection*{Comparison between \MultiBags and \MultiBagsPlus}

Next, we compare the performance difference between \MultiBags and
\MultiBagsPlus.  To evaluate the overhead difference between them, we run the
same programs (i.e., with structured futures) with both algorithms.  Although
\MultiBagsPlus is designed for general futures, it also works with programs
that use structured futures, albeit with an additional $k^2$ overhead, where
$k$ is the number of \getF calls.

For \texttt{lcs}, \texttt{sw}, and \texttt{mm}, $k$ is dictated by how
much the base case is coarsened --- the smaller the base case, the
more \getF calls, and the higher $k$ is (which leads to higher
overhead).  Runtimes shown before used base case of $B = \sqrt{N}$ to keep 
the work asymptotically the same across baseline, \MultiBags, and  
\MultiBagsPlus.  Now we decrease the base case size below
(i.e., increase $k$) to see how the overhead of \MultiBagsPlus changes 
compared with the overhead of \MultiBags.

\begin{figure}[h]
\vspace{-2mm}
\centering
\begin{small}
\setlength\tabcolsep{5pt}
\newcommand{\oh}[1]{\hspace{1pt}\scriptsize (#1$\times$)}
\newcommand{\cheader}[3]{\multicolumn{#1}{#2}{\textit{#3}}}
\begin{tabular}{rrr@{\hspace{1pt}}rr@{\hspace{1pt}}r}
    & & \multicolumn{4}{c}{\it reachability} \\
    \cheader{1}{c}{bench} & \cheader{1}{c}{baseline} 
    & \cheader{2}{c}{\MultiBags} & \cheader{2}{c}{\MultiBagsPlus} \\
\hline
\texttt{lcs} \texttt{(B=64)} &  2.14 &  2.20 & \oh{1.03} &  4.68 & \oh{2.19} \\
\texttt{lcs} \texttt{(B=32)} &  2.14 &  2.09 & \oh{0.98} & 39.82 & \oh{18.63} \\
\hline
\texttt{sw}  \texttt{(B=32)} & 14.57 & 14.69 & \oh{1.01} & 13.97 & \oh{0.96} \\
\hline
\texttt{mm}  \texttt{(B=32)} & 13.08 & 13.12 & \oh{1.00} & 49.11 & \oh{3.75} \\
\end{tabular}
\vspace{-2mm}
\caption{The execution times under the baseline and reachability configurations
    (both \MultiBags and \MultiBagsPlus) for a subset of benchmarks
    implemented with structured futures.  Numbers in the parentheses show 
    the overhead compared to the baseline.}
\label{fig:perf-diff}
\end{small}
\end{figure}


\figref{perf-diff} shows the measurements for running programs with
structured futures using \MultiBags and \MultiBagsPlus in the reachability
configuration with different base cases.  The overhead difference between 
\MultiBags and \MultiBagsPlus can readily be observed in 
\figreftwo{structured-overhead}{general-overhead} --- compared to \MultiBags, 
    \MultiBagsPlus incurs $2+\times$ more overhead running \texttt{dedup} 
    and $3+\times$ more running \texttt{bst} for maintaining reachability.
Here we show additional numbers for benchmarks where varying base case 
sizes changes $k$.


The measurements with \texttt{lcs} and \texttt{mm} bear out the extra
overhead of \MultiBagsPlus.  The \texttt{lcs} benchmark has $\Theta(n^2)$ work
versus $\left(n/B\right)^2$ futures, while \texttt{mm} has more work
($\Theta(n^3)$), but also requires $(n/B)^3$ futures.  With a higher ratio of
futures to total work, the overhead is more apparent.  Moreover, the memory
required for the reachability matrix $\cal R$ becomes substantial for small
base cases, adding more overhead. The \texttt{sw} benchmark, however, has
$\Theta(n^3)$ work compared to $\left(n/B\right)^2$ futures, so the effect of
smaller base cases is small.

\secput{related}{Related Work}




Besides works discussed in \secref{intro}, researchers have considered
race detection for other structured computations. \citet{DimitrovVeSa15}
propose a sequential near-optimal race detection algorithm for two-dimensional
dags which also exhibit nice structural properties.  Subsequently,
\citet{XuLeAg18} propose a race detector for two-dimensional dags with
asymptotically optimal parallel running time.  \citet{LeeSc15} propose a
sequential race detector for fork-join computations with reductions, where the
computation dag is almost series-parallel except when reductions are
performed.



Beyond task parallel code, there is a rich literature on race detection for
programming models that generate nondeterministic computations, such as ones
that employ persistent threads and locks.  For such models, since the dag
necessarily depends on the schedule, the best correctness guarantee that a
race detector can provide is for a given program, for a given input, and for a
given \emph{schedule}.  Early work~\cite{SavageBuNe97,VonPraunGr01} employs
lock-set algorithm, which provides wide coverage but can lead to many false
positives, because it cannot precisely capture happens-before (HB) relations
formed between threads.  

A vector-clock (VC) based algorithm such as one proposed
by~\citet{FlanaganFr09} can capture HB precisely for a given schedule.  Such
algorithm can be used on computation with arbitrary dependences, but naively
applying it to task parallel code would be impractical, since it requires
storing a VC of length $n$ with each each memory location querying against it
per access, incurring a multiplicative factor of $n$ overhead on top of the
work, where $n$ is the number of strands, which can be on the order of
millions.

In the context of race detecting nondeterministic code, researchers have
investigated hybrid approaches incorporating VC and
lock-set~\cite{OCallahanCh03, PoznianskySc03, YuRoCh05, SerebryanyIs09} to
trade-off precisions and coverage.  More recently, researchers have proposed
predictive analysis to explore alternative feasible schedules among close by
instructions to increase the coverage (e.g.~\cite{SmaragdakisEvSa12, SaidWa11,
LiuTrZh16, KiniMaVi17}) while keeping the precision.

\punt{ 
Some other race detectors, such as FastTrack~\cite{FlanaganFr09}, do not
differentiate between nested parallelism and unstructured parallelism. For a
program with sequential running time $T_1$, $n$ nested parallel constructs,
and $k$ arbitrary additional edges, FastTrack achieves a running time of
$O(T_1 + (n+k)^2)$.  Our algorithm would perform similar to FastTrack's if $k
= \Omega(n)$, but is much better if most of the parallelism in the program can
be expressed through nested parallelism.  }

\section*{Acknowledgements}
This research was supported in part by National Science Foundation
grants CCF-1150036, CCF-1527692, CCF-1733873, and XPS-1439062. We
thank the referees and our shepherd for their excellent comments.

\clearpage

\thispagestyle{empty}

\begin{thebibliography}{60}


\ifx \showCODEN    \undefined \def \showCODEN     #1{\unskip}     \fi
\ifx \showDOI      \undefined \def \showDOI       #1{#1}\fi
\ifx \showISBNx    \undefined \def \showISBNx     #1{\unskip}     \fi
\ifx \showISBNxiii \undefined \def \showISBNxiii  #1{\unskip}     \fi
\ifx \showISSN     \undefined \def \showISSN      #1{\unskip}     \fi
\ifx \showLCCN     \undefined \def \showLCCN      #1{\unskip}     \fi
\ifx \shownote     \undefined \def \shownote      #1{#1}          \fi
\ifx \showarticletitle \undefined \def \showarticletitle #1{#1}   \fi
\ifx \showURL      \undefined \def \showURL       {\relax}        \fi
\providecommand\bibfield[2]{#2}
\providecommand\bibinfo[2]{#2}
\providecommand\natexlab[1]{#1}
\providecommand\showeprint[2][]{arXiv:#2}

\bibitem[\protect\citeauthoryear{Agrawal, Devietti, Fineman, Lee, Utterback,
  and Xu}{Agrawal et~al\mbox{.}}{2018}]%
        {AgrawalDeFi18}
\bibfield{author}{\bibinfo{person}{Kunal Agrawal}, \bibinfo{person}{Joseph
  Devietti}, \bibinfo{person}{Jeremy~T. Fineman},
  \bibinfo{person}{I-Ting~Angelina Lee}, \bibinfo{person}{Robert Utterback},
  {and} \bibinfo{person}{Changming Xu}.} \bibinfo{year}{2018}\natexlab{}.
\newblock \showarticletitle{Race Detection and Reachability in Nearly
  Series-parallel DAGs}. In \bibinfo{booktitle}{{\em Proceedings of the
  Twenty-Ninth Annual ACM-SIAM Symposium on Discrete Algorithms}} {\em
  (\bibinfo{series}{SODA '18})}. \bibinfo{publisher}{Society for Industrial and
  Applied Mathematics}, \bibinfo{address}{New Orleans, Louisiana},
  \bibinfo{pages}{156--171}.
\newblock
\showISBNx{978-1-6119-7503-1}
\showURL{%
\url{http://dl.acm.org/citation.cfm?id=3174304.3175277}}


\bibitem[\protect\citeauthoryear{Arora, Blumofe, and Plaxton}{Arora
  et~al\mbox{.}}{1998}]%
        {AroraBlPl98}
\bibfield{author}{\bibinfo{person}{Nimar~S. Arora}, \bibinfo{person}{Robert~D.
  Blumofe}, {and} \bibinfo{person}{C.~Greg Plaxton}.}
  \bibinfo{year}{1998}\natexlab{}.
\newblock \showarticletitle{Thread Scheduling for Multiprogrammed
  Multiprocessors}. In \bibinfo{booktitle}{{\em 10th Annual ACM Symposium on
  Parallel Algorithms and Architectures}}. \bibinfo{pages}{119--129}.
\newblock


\bibitem[\protect\citeauthoryear{Arvind, Nikhil, and Pingali}{Arvind
  et~al\mbox{.}}{1986}]%
        {ArvindNiPi86}
\bibfield{author}{\bibinfo{person}{Arvind}, \bibinfo{person}{R.S. Nikhil},
  {and} \bibinfo{person}{K.K. Pingali}.} \bibinfo{year}{1986}\natexlab{}.
\newblock \showarticletitle{{I-structures: Data Structures for Parallel
  Computing}}. In \bibinfo{booktitle}{{\em Proceedings of the Graph Reduction
  Workshop}}.
\newblock


\bibitem[\protect\citeauthoryear{Baker and Hewitt}{Baker and Hewitt}{1977}]%
        {BakerHe77}
\bibfield{author}{\bibinfo{person}{Henry~C. Baker, Jr.} {and}
  \bibinfo{person}{Carl Hewitt}.} \bibinfo{year}{1977}\natexlab{}.
\newblock \showarticletitle{The incremental garbage collection of processes}.
\newblock \bibinfo{journal}{{\em SIGPLAN Notices\/}} \bibinfo{volume}{12},
  \bibinfo{number}{8} (\bibinfo{year}{1977}), \bibinfo{pages}{55--59}.
\newblock


\bibitem[\protect\citeauthoryear{Barik, Budimli{\'{c}}, Cav{\`{e}}, Chatterjee,
  Guo, Peixotto, Raman, Shirako, Ta{\c{s}}{\i}rlar, Yan, Zhao, and
  Sarkar}{Barik et~al\mbox{.}}{2009}]%
        {BarikBuCa09}
\bibfield{author}{\bibinfo{person}{Rajkishore Barik}, \bibinfo{person}{Zoran
  Budimli{\'{c}}}, \bibinfo{person}{Vincent Cav{\`{e}}},
  \bibinfo{person}{Sanjay Chatterjee}, \bibinfo{person}{Yi Guo},
  \bibinfo{person}{David Peixotto}, \bibinfo{person}{Raghavan Raman},
  \bibinfo{person}{Jun Shirako}, \bibinfo{person}{Sa{\u{g}}nak
  Ta{\c{s}}{\i}rlar}, \bibinfo{person}{Yonghong Yan}, \bibinfo{person}{Yisheng
  Zhao}, {and} \bibinfo{person}{Vivek Sarkar}.}
  \bibinfo{year}{2009}\natexlab{}.
\newblock \showarticletitle{The {Habanero} Multicore Software Research
  Project}. In \bibinfo{booktitle}{{\em Proceedings of the 24th ACM SIGPLAN
  Conference Companion on Object Oriented Programming Systems Languages and
  Applications}} {\em (\bibinfo{series}{OOPSLA '09})}.
  \bibinfo{publisher}{ACM}, \bibinfo{address}{Orlando, Florida, USA},
  \bibinfo{pages}{735--736}.
\newblock
\showISBNx{978-1-60558-768-4}


\bibitem[\protect\citeauthoryear{Bender, Fineman, Gilbert, and
  Leiserson}{Bender et~al\mbox{.}}{2004}]%
        {BenderFiGi04}
\bibfield{author}{\bibinfo{person}{Michael~A. Bender},
  \bibinfo{person}{Jeremy~T. Fineman}, \bibinfo{person}{Seth Gilbert}, {and}
  \bibinfo{person}{Charles~E. Leiserson}.} \bibinfo{year}{2004}\natexlab{}.
\newblock \showarticletitle{On-the-Fly Maintenance of Series-Parallel
  Relationships in Fork-Join Multithreaded Programs}. In
  \bibinfo{booktitle}{{\em 16th Annual ACM Symposium on Parallel Algorithms and
  Architectures}}. \bibinfo{pages}{133--144}.
\newblock


\bibitem[\protect\citeauthoryear{Bienia, Kumar, Singh, and Li}{Bienia
  et~al\mbox{.}}{2008}]%
        {BieniaKuSi08}
\bibfield{author}{\bibinfo{person}{Christian Bienia}, \bibinfo{person}{Sanjeev
  Kumar}, \bibinfo{person}{Jaswinder~Pal Singh}, {and} \bibinfo{person}{Kai
  Li}.} \bibinfo{year}{2008}\natexlab{}.
\newblock \showarticletitle{The {PARSEC} Benchmark Suite: Characterization and
  Architectural Implications}. In \bibinfo{booktitle}{{\em PACT}}.
  \bibinfo{publisher}{ACM}, \bibinfo{pages}{72--81}.
\newblock


\bibitem[\protect\citeauthoryear{Bienia and Li}{Bienia and Li}{2010}]%
        {BieniaLi10}
\bibfield{author}{\bibinfo{person}{Christian Bienia} {and} \bibinfo{person}{Kai
  Li}.} \bibinfo{year}{2010}\natexlab{}.
\newblock \showarticletitle{Characteristics of Workloads Using the Pipeline
  Programming Model}. In \bibinfo{booktitle}{{\em ISCA}}.
  \bibinfo{publisher}{Springer-Verlag}, \bibinfo{pages}{161--171}.
\newblock


\bibitem[\protect\citeauthoryear{Blelloch, Gibbons, Matias, and
  Narlikar}{Blelloch et~al\mbox{.}}{1997}]%
        {BlellochGiMa97}
\bibfield{author}{\bibinfo{person}{Guy~E. Blelloch},
  \bibinfo{person}{Phillip~B. Gibbons}, \bibinfo{person}{Yossi Matias}, {and}
  \bibinfo{person}{Girija~J. Narlikar}.} \bibinfo{year}{1997}\natexlab{}.
\newblock \showarticletitle{Space-Efficient Scheduling of Parallelism with
  Synchronization Variables}. In \bibinfo{booktitle}{{\em 9th Annual ACM
  Symposium on Parallel Algorithms and Architectures}}.
  \bibinfo{pages}{12--23}.
\newblock


\bibitem[\protect\citeauthoryear{Blelloch and Reid-Miller}{Blelloch and
  Reid-Miller}{1997}]%
        {BlellochRe97}
\bibfield{author}{\bibinfo{person}{Guy~E. Blelloch} {and}
  \bibinfo{person}{Margaret Reid-Miller}.} \bibinfo{year}{1997}\natexlab{}.
\newblock \showarticletitle{Pipelining with futures}. In
  \bibinfo{booktitle}{{\em SPAA}}. \bibinfo{publisher}{ACM},
  \bibinfo{pages}{249--259}.
\newblock
\showISBNx{0-89791-890-8}


\bibitem[\protect\citeauthoryear{Budimli{\'{c}}, Burke, Cav{\'e}, Knobe,
  Lowney, Newton, Palsberg, Peixotto, Sarkar, Schlimbach, and
  Ta{\c{s}\i}rlar}{Budimli{\'{c}} et~al\mbox{.}}{2010}]%
        {BudimlicBuCa10}
\bibfield{author}{\bibinfo{person}{Zoran Budimli{\'{c}}},
  \bibinfo{person}{Michael Burke}, \bibinfo{person}{Vincent Cav{\'e}},
  \bibinfo{person}{Kathleen Knobe}, \bibinfo{person}{Geoff Lowney},
  \bibinfo{person}{Ryan Newton}, \bibinfo{person}{Jens Palsberg},
  \bibinfo{person}{David Peixotto}, \bibinfo{person}{Vivek Sarkar},
  \bibinfo{person}{Frank Schlimbach}, {and} \bibinfo{person}{Sa{\u{g}}nak
  Ta{\c{s}\i}rlar}.} \bibinfo{year}{2010}\natexlab{}.
\newblock \showarticletitle{Concurrent Collections}.
\newblock \bibinfo{journal}{{\em Journal of Scientific Programming\/}}
  \bibinfo{volume}{18}, \bibinfo{number}{3-4} (\bibinfo{date}{Aug.}
  \bibinfo{year}{2010}), \bibinfo{pages}{203--217}.
\newblock
\showISSN{1058-9244}


\bibitem[\protect\citeauthoryear{Cav{\'e}, Zhao, Shirako, and Sarkar}{Cav{\'e}
  et~al\mbox{.}}{2011}]%
        {CaveZhSh11}
\bibfield{author}{\bibinfo{person}{Vincent Cav{\'e}}, \bibinfo{person}{Jisheng
  Zhao}, \bibinfo{person}{Jun Shirako}, {and} \bibinfo{person}{Vivek Sarkar}.}
  \bibinfo{year}{2011}\natexlab{}.
\newblock \showarticletitle{{H}abanero-{J}ava: the new adventures of old
  {X10}}. In \bibinfo{booktitle}{{\em Proceedings of the 9th International
  Conference on Principles and Practice of Programming in Java}} {\em
  (\bibinfo{series}{PPPJ '11})}. \bibinfo{pages}{51--61}.
\newblock


\bibitem[\protect\citeauthoryear{Chandra, Gupta, and Hennessy}{Chandra
  et~al\mbox{.}}{1994}]%
        {ChandraGuHe94}
\bibfield{author}{\bibinfo{person}{Rohit Chandra}, \bibinfo{person}{Anoop
  Gupta}, {and} \bibinfo{person}{John~L. Hennessy}.}
  \bibinfo{year}{1994}\natexlab{}.
\newblock \showarticletitle{{COOL}: An Object-Based Language for Parallel
  Programming}.
\newblock \bibinfo{journal}{{\em IEEE Computer\/}} \bibinfo{volume}{27},
  \bibinfo{number}{8} (\bibinfo{date}{Aug.} \bibinfo{year}{1994}),
  \bibinfo{pages}{13--26}.
\newblock


\bibitem[\protect\citeauthoryear{Charles, Grothoff, Saraswat, Donawa, Kielstra,
  Ebcioglu, von Praun, and Sarkar}{Charles et~al\mbox{.}}{2005}]%
        {CharlesGrSa+05}
\bibfield{author}{\bibinfo{person}{Philippe Charles},
  \bibinfo{person}{Christian Grothoff}, \bibinfo{person}{Vijay Saraswat},
  \bibinfo{person}{Christopher Donawa}, \bibinfo{person}{Allan Kielstra},
  \bibinfo{person}{Kemal Ebcioglu}, \bibinfo{person}{Christoph von Praun},
  {and} \bibinfo{person}{Vivek Sarkar}.} \bibinfo{year}{2005}\natexlab{}.
\newblock \showarticletitle{X10: An Object-Oriented Approach to Non-Uniform
  Cluster Computing}. In \bibinfo{booktitle}{{\em 20th Annual ACM SIGPLAN
  Conference on Object-Oriented Programming, Systems, Languages, and
  Applications}}. \bibinfo{pages}{519--538}.
\newblock


\bibitem[\protect\citeauthoryear{Che, Boyer, Meng, Tarjan, Sheaffer, Lee, and
  Skadron}{Che et~al\mbox{.}}{2009}]%
        {CheBoMe09}
\bibfield{author}{\bibinfo{person}{Shuai Che}, \bibinfo{person}{Michael Boyer},
  \bibinfo{person}{Jiayuan Meng}, \bibinfo{person}{David Tarjan},
  \bibinfo{person}{Jeremy~W. Sheaffer}, \bibinfo{person}{Sang-Ha Lee}, {and}
  \bibinfo{person}{Kevin Skadron}.} \bibinfo{year}{2009}\natexlab{}.
\newblock \showarticletitle{Rodinia: A benchmark suite for heterogeneous
  computing}. In \bibinfo{booktitle}{{\em 2009 IEEE International Symposium on
  Workload Characterization (IISWC)}}. \bibinfo{pages}{44--54}.
\newblock


\bibitem[\protect\citeauthoryear{Cormen, Leiserson, Rivest, and Stein}{Cormen
  et~al\mbox{.}}{2009}]%
        {CormenLeRi09}
\bibfield{author}{\bibinfo{person}{Thomas~H. Cormen},
  \bibinfo{person}{Charles~E. Leiserson}, \bibinfo{person}{Ronald~L. Rivest},
  {and} \bibinfo{person}{Clifford Stein}.} \bibinfo{year}{2009}\natexlab{}.
\newblock \bibinfo{booktitle}{{\em Introduction to Algorithms\/}
  (\bibinfo{edition}{third} ed.)}.
\newblock \bibinfo{publisher}{The MIT Press}.
\newblock


\bibitem[\protect\citeauthoryear{Danaher, Lee, and Leiserson}{Danaher
  et~al\mbox{.}}{2008}]%
        {DanaherLeLe06}
\bibfield{author}{\bibinfo{person}{John~S. Danaher},
  \bibinfo{person}{I-Ting~Angelina Lee}, {and} \bibinfo{person}{Charles~E.
  Leiserson}.} \bibinfo{year}{2008}\natexlab{}.
\newblock \showarticletitle{Programming with exceptions in {JCilk}}.
\newblock \bibinfo{journal}{{\em Science of Computer Programming\/}}
  \bibinfo{volume}{63}, \bibinfo{number}{2} (\bibinfo{date}{Dec.}
  \bibinfo{year}{2008}), \bibinfo{pages}{147--171}.
\newblock


\bibitem[\protect\citeauthoryear{Dimitrov, Vechev, and Sarkar}{Dimitrov
  et~al\mbox{.}}{2015}]%
        {DimitrovVeSa15}
\bibfield{author}{\bibinfo{person}{Dimitar Dimitrov}, \bibinfo{person}{Martin
  Vechev}, {and} \bibinfo{person}{Vivek Sarkar}.}
  \bibinfo{year}{2015}\natexlab{}.
\newblock \showarticletitle{Race Detection in Two Dimensions}. In
  \bibinfo{booktitle}{{\em Proceedings of the 27th ACM Symposium on Parallelism
  in Algorithms and Architectures}} {\em (\bibinfo{series}{SPAA '15})}.
  \bibinfo{publisher}{ACM}, \bibinfo{address}{Portland, Oregon, USA},
  \bibinfo{pages}{101--110}.
\newblock
\showISBNx{978-1-4503-3588-1}
\showDOI{%
\url{https://doi.org/10.1145/2755573.2755601}}


\bibitem[\protect\citeauthoryear{Feng and Leiserson}{Feng and
  Leiserson}{1997}]%
        {FengLe97}
\bibfield{author}{\bibinfo{person}{Mingdong Feng} {and}
  \bibinfo{person}{Charles~E. Leiserson}.} \bibinfo{year}{1997}\natexlab{}.
\newblock \showarticletitle{Efficient Detection of Determinacy Races in {Cilk}
  Programs}. In \bibinfo{booktitle}{{\em Proceedings of the Ninth Annual ACM
  Symposium on Parallel Algorithms and Architectures (SPAA)}}.
  \bibinfo{pages}{1--11}.
\newblock


\bibitem[\protect\citeauthoryear{Feng and Leiserson}{Feng and
  Leiserson}{1999}]%
        {FengLe99}
\bibfield{author}{\bibinfo{person}{Mingdong Feng} {and}
  \bibinfo{person}{Charles~E. Leiserson}.} \bibinfo{year}{1999}\natexlab{}.
\newblock \showarticletitle{Efficient Detection of Determinacy Races in {Cilk}
  Programs}.
\newblock \bibinfo{journal}{{\em Theory of Computing Systems\/}}
  \bibinfo{volume}{32}, \bibinfo{number}{3} (\bibinfo{year}{1999}),
  \bibinfo{pages}{301--326}.
\newblock


\bibitem[\protect\citeauthoryear{Fineman}{Fineman}{2005}]%
        {Fineman05}
\bibfield{author}{\bibinfo{person}{Jeremy~T. Fineman}.}
  \bibinfo{year}{2005}\natexlab{}.
\newblock {\em \bibinfo{title}{Provably Good Race Detection That Runs in
  Parallel}}.
\newblock \bibinfo{thesistype}{Master's\ thesis}.
  \bibinfo{school}{Massachusetts Institute of Technology, Department of
  Electrical Engineering and Computer Science}, \bibinfo{address}{Cambridge,
  MA}.
\newblock


\bibitem[\protect\citeauthoryear{Flanagan and Freund}{Flanagan and
  Freund}{2009}]%
        {FlanaganFr09}
\bibfield{author}{\bibinfo{person}{Cormac Flanagan} {and}
  \bibinfo{person}{Stephen~N. Freund}.} \bibinfo{year}{2009}\natexlab{}.
\newblock \showarticletitle{FastTrack: efficient and precise dynamic race
  detection}.
\newblock \bibinfo{journal}{{\em SIGPLAN Not.\/}} \bibinfo{volume}{44},
  \bibinfo{number}{6} (\bibinfo{date}{June} \bibinfo{year}{2009}),
  \bibinfo{pages}{121--133}.
\newblock
\showISSN{0362-1340}


\bibitem[\protect\citeauthoryear{Fluet, Rainey, Reppy, and Shaw}{Fluet
  et~al\mbox{.}}{2010}]%
        {FluetRaRe10}
\bibfield{author}{\bibinfo{person}{Matthew Fluet}, \bibinfo{person}{Mike
  Rainey}, \bibinfo{person}{John Reppy}, {and} \bibinfo{person}{Adam Shaw}.}
  \bibinfo{year}{2010}\natexlab{}.
\newblock \showarticletitle{Implicitly Threaded Parallelism in Manticore}.
\newblock \bibinfo{journal}{{\em Journal of Functional Programming\/}}
  \bibinfo{volume}{20}, \bibinfo{number}{5-6} (\bibinfo{date}{Nov.}
  \bibinfo{year}{2010}), \bibinfo{pages}{537--576}.
\newblock
\showISSN{0956-7968}
\showDOI{%
\url{https://doi.org/10.1017/S0956796810000201}}


\bibitem[\protect\citeauthoryear{Friedman and Wise}{Friedman and Wise}{1978}]%
        {FriedmanWi78}
\bibfield{author}{\bibinfo{person}{D.P. Friedman} {and} \bibinfo{person}{D.S.
  Wise}.} \bibinfo{year}{1978}\natexlab{}.
\newblock \showarticletitle{Aspects of Applicative Programming for Parallel
  Processing}.
\newblock \bibinfo{journal}{{\it IEEE Trans. Comput.}} \bibinfo{volume}{C-27},
  \bibinfo{number}{4} (\bibinfo{year}{1978}), \bibinfo{pages}{289--296}.
\newblock


\bibitem[\protect\citeauthoryear{Frigo, Leiserson, and Randall}{Frigo
  et~al\mbox{.}}{1998}]%
        {FrigoLeRa98}
\bibfield{author}{\bibinfo{person}{Matteo Frigo}, \bibinfo{person}{Charles~E.
  Leiserson}, {and} \bibinfo{person}{Keith~H. Randall}.}
  \bibinfo{year}{1998}\natexlab{}.
\newblock \showarticletitle{The Implementation of the {C}ilk-5 Multithreaded
  Language}. In \bibinfo{booktitle}{{\em PLDI}}. ACM,
  \bibinfo{pages}{212--223}.
\newblock


\bibitem[\protect\citeauthoryear{Halstead}{Halstead}{1985}]%
        {Halstead85}
\bibfield{author}{\bibinfo{person}{Robert~H. Halstead, Jr.}}
  \bibinfo{year}{1985}\natexlab{}.
\newblock \showarticletitle{Multilisp: A Language for Concurrent Symbolic
  Computation}.
\newblock \bibinfo{journal}{{\em ACM TOPLAS\/}} \bibinfo{volume}{7},
  \bibinfo{number}{4} (\bibinfo{date}{Oct.} \bibinfo{year}{1985}),
  \bibinfo{pages}{501--538}.
\newblock


\bibitem[\protect\citeauthoryear{Herlihy and Liu}{Herlihy and Liu}{2014}]%
        {HerlihyLi14}
\bibfield{author}{\bibinfo{person}{Maurice Herlihy} {and}
  \bibinfo{person}{Zhiyu Liu}.} \bibinfo{year}{2014}\natexlab{}.
\newblock \showarticletitle{Well-structured Futures and Cache Locality}. In
  \bibinfo{booktitle}{{\em Proceedings of the 19th ACM SIGPLAN Symposium on
  Principles and Practice of Parallel Programming}} {\em
  (\bibinfo{series}{PPoPP '14})}. \bibinfo{publisher}{ACM},
  \bibinfo{address}{Orlando, Florida, USA}, \bibinfo{pages}{155--166}.
\newblock
\showISBNx{978-1-4503-2656-8}
\showDOI{%
\url{https://doi.org/10.1145/2555243.2555257}}


\bibitem[\protect\citeauthoryear{Intel Corporation}{Intel}{2013}]%
        {IntelCilkPlus13}
Intel \bibinfo{year}{2013}\natexlab{}.
\newblock \bibinfo{title}{Intel\textregistered\
  {Cilk\textsuperscript{\texttrademark}} {Plus}}.
\newblock \bibinfo{howpublished}{\url{https://www.cilkplus.org}}.
  (\bibinfo{year}{2013}).
\newblock


\bibitem[\protect\citeauthoryear{Intel Corporation}{Intel Corporation}{2012}]%
        {IntelTBBManual}
Intel Corporation \bibinfo{year}{2012}\natexlab{}.
\newblock \bibinfo{booktitle}{{\em Intel(R) Threading Building Blocks}}.
\newblock Intel Corporation.
\newblock
\newblock
\shownote{Available from
  \url{http://software.intel.com/sites/products/documentation/doclib/tbb_sa/help/index.htm}.}


\bibitem[\protect\citeauthoryear{Intel Corporation}{Intel Corporation}{2013}]%
        {IntelCilkPlusLangSpec13}
Intel Corporation \bibinfo{year}{2013}\natexlab{}.
\newblock \bibinfo{booktitle}{{\em Intel\textregistered\
  {Cilk\textsuperscript{\texttrademark}} {Plus} Language Extension
  Specification, Version 1.1}}.
\newblock Intel Corporation.
\newblock
\newblock
\shownote{Document 324396-002US. Available from
  \url{http://cilkplus.org/sites/default/files/open_specifications/Intel_Cilk_plus_lang_spec_2.htm}.}


\bibitem[\protect\citeauthoryear{Kini, Mathur, and Viswanathan}{Kini
  et~al\mbox{.}}{2017}]%
        {KiniMaVi17}
\bibfield{author}{\bibinfo{person}{Dileep Kini}, \bibinfo{person}{Umang
  Mathur}, {and} \bibinfo{person}{Mahesh Viswanathan}.}
  \bibinfo{year}{2017}\natexlab{}.
\newblock \showarticletitle{Dynamic Race Prediction in Linear Time}. In
  \bibinfo{booktitle}{{\em Proceedings of the 38th ACM SIGPLAN Conference on
  Programming Language Design and Implementation}} {\em (\bibinfo{series}{PLDI
  2017})}. \bibinfo{publisher}{ACM}, \bibinfo{address}{New York, NY, USA},
  \bibinfo{pages}{157--170}.
\newblock
\showISBNx{978-1-4503-4988-8}
\showDOI{%
\url{https://doi.org/10.1145/3062341.3062374}}


\bibitem[\protect\citeauthoryear{Kogan and Herlihy}{Kogan and Herlihy}{2014}]%
        {KoganHe14}
\bibfield{author}{\bibinfo{person}{Alex Kogan} {and} \bibinfo{person}{Maurice
  Herlihy}.} \bibinfo{year}{2014}\natexlab{}.
\newblock \showarticletitle{The Future(s) of Shared Data Structures}. In
  \bibinfo{booktitle}{{\em Proceedings of the 2014 ACM Symposium on Principles
  of Distributed Computing}} {\em (\bibinfo{series}{PODC '14})}.
  \bibinfo{publisher}{ACM}, \bibinfo{address}{Paris, France},
  \bibinfo{pages}{30--39}.
\newblock
\showISBNx{978-1-4503-2944-6}
\showURL{%
\url{http://doi.acm.org/10.1145/2611462.2611496}}


\bibitem[\protect\citeauthoryear{Kranz, Halstead, and Mohr}{Kranz
  et~al\mbox{.}}{1989}]%
        {KranzHaMo89}
\bibfield{author}{\bibinfo{person}{David~A. Kranz}, \bibinfo{person}{Robert~H.
  Halstead, Jr.}, {and} \bibinfo{person}{Eric Mohr}.}
  \bibinfo{year}{1989}\natexlab{}.
\newblock \showarticletitle{{Mul-T}: A High-Performance Parallel {Lisp}}. In
  \bibinfo{booktitle}{{\em PLDI}}. \bibinfo{publisher}{ACM},
  \bibinfo{pages}{81--90}.
\newblock


\bibitem[\protect\citeauthoryear{Lee and Schardl}{Lee and Schardl}{2015}]%
        {LeeSc15}
\bibfield{author}{\bibinfo{person}{I-Ting~Angelina Lee} {and}
  \bibinfo{person}{Tao~B. Schardl}.} \bibinfo{year}{2015}\natexlab{}.
\newblock \showarticletitle{Efficiently Detecting Races in Cilk Programs That
  Use Reducer Hyperobjects}. In \bibinfo{booktitle}{{\em SPAA '15: Proceedings
  of the 27th ACM on Symposium on Parallelism in Algorithms and Architectures}}
  {\em (\bibinfo{series}{SPAA '15})}. \bibinfo{publisher}{ACM},
  \bibinfo{address}{Portland, Oregon, USA}, \bibinfo{pages}{111--122}.
\newblock
\showISBNx{978-1-4503-3588-1}
\showURL{%
\url{http://doi.acm.org/10.1145/2755573.2755599}}


\bibitem[\protect\citeauthoryear{Leiserson}{Leiserson}{2010}]%
        {Leiserson10}
\bibfield{author}{\bibinfo{person}{Charles~E. Leiserson}.}
  \bibinfo{year}{2010}\natexlab{}.
\newblock \showarticletitle{The {Cilk++} Concurrency Platform}.
\newblock \bibinfo{journal}{{\em J. Supercomputing\/}} \bibinfo{volume}{51},
  \bibinfo{number}{3} (\bibinfo{year}{2010}), \bibinfo{pages}{244--257}.
\newblock


\bibitem[\protect\citeauthoryear{Liu, Tripp, and Zhang}{Liu
  et~al\mbox{.}}{2016}]%
        {LiuTrZh16}
\bibfield{author}{\bibinfo{person}{Peng Liu}, \bibinfo{person}{Omer Tripp},
  {and} \bibinfo{person}{Xiangyu Zhang}.} \bibinfo{year}{2016}\natexlab{}.
\newblock \showarticletitle{IPA: Improving Predictive Analysis with Pointer
  Analysis}. In \bibinfo{booktitle}{{\em Proceedings of the 25th International
  Symposium on Software Testing and Analysis}} {\em (\bibinfo{series}{ISSTA
  2016})}. \bibinfo{publisher}{ACM}, \bibinfo{address}{New York, NY, USA},
  \bibinfo{pages}{59--69}.
\newblock
\showISBNx{978-1-4503-4390-9}
\showDOI{%
\url{https://doi.org/10.1145/2931037.2931046}}


\bibitem[\protect\citeauthoryear{Lu, Ji, and Scott}{Lu et~al\mbox{.}}{2014}]%
        {LuJiSc14}
\bibfield{author}{\bibinfo{person}{Li Lu}, \bibinfo{person}{Weixing Ji}, {and}
  \bibinfo{person}{Michael~L. Scott}.} \bibinfo{year}{2014}\natexlab{}.
\newblock \showarticletitle{Dynamic Enforcement of Determinism in a Parallel
  Scripting Language}. In \bibinfo{booktitle}{{\em Proceedings of the 35th ACM
  SIGPLAN Conference on Programming Language Design and Implementation}} {\em
  (\bibinfo{series}{PLDI '14})}. \bibinfo{publisher}{ACM},
  \bibinfo{address}{Edinburgh, United Kingdom}, \bibinfo{pages}{519--529}.
\newblock
\showISBNx{978-1-4503-2784-8}
\showDOI{%
\url{https://doi.org/10.1145/2594291.2594300}}


\bibitem[\protect\citeauthoryear{Mellor-Crummey}{Mellor-Crummey}{1991}]%
        {Mellor-Crummey91}
\bibfield{author}{\bibinfo{person}{John Mellor-Crummey}.}
  \bibinfo{year}{1991}\natexlab{}.
\newblock \showarticletitle{On-the-fly Detection of Data Races for Programs
  with Nested Fork-Join Parallelism}. In \bibinfo{booktitle}{{\em Proceedings
  of Supercomputing'91}}. \bibinfo{pages}{24--33}.
\newblock


\bibitem[\protect\citeauthoryear{Netzer and Miller}{Netzer and Miller}{1992}]%
        {NetzerMi92}
\bibfield{author}{\bibinfo{person}{Robert H.~B. Netzer} {and}
  \bibinfo{person}{Barton~P. Miller}.} \bibinfo{year}{1992}\natexlab{}.
\newblock \showarticletitle{What are Race Conditions?}
\newblock \bibinfo{journal}{{\em ACM Letters on Programming Languages and
  Systems\/}} \bibinfo{volume}{1}, \bibinfo{number}{1} (\bibinfo{date}{March}
  \bibinfo{year}{1992}), \bibinfo{pages}{74--88}.
\newblock


\bibitem[\protect\citeauthoryear{O'Callahan and Choi}{O'Callahan and
  Choi}{2003}]%
        {OCallahanCh03}
\bibfield{author}{\bibinfo{person}{Robert O'Callahan} {and}
  \bibinfo{person}{Jong-Deok Choi}.} \bibinfo{year}{2003}\natexlab{}.
\newblock \showarticletitle{Hybrid Dynamic Data Race Detection}. In
  \bibinfo{booktitle}{{\em Proceedings of the Ninth ACM SIGPLAN Symposium on
  Principles and Practice of Parallel Programming}} {\em
  (\bibinfo{series}{PPoPP '03})}. \bibinfo{publisher}{ACM},
  \bibinfo{address}{New York, NY, USA}, \bibinfo{pages}{167--178}.
\newblock


\bibitem[\protect\citeauthoryear{OpenMP 4.0}{OpenMP 4.0}{2013}]%
        {OpenMP13}
OpenMP 4.0 \bibinfo{year}{2013}\natexlab{}.
\newblock \bibinfo{booktitle}{{\em {OpenMP} Application Program Interface,
  Version 4.0}}.
\newblock


\bibitem[\protect\citeauthoryear{Pozniansky and Schuster}{Pozniansky and
  Schuster}{2003}]%
        {PoznianskySc03}
\bibfield{author}{\bibinfo{person}{Eli Pozniansky} {and} \bibinfo{person}{Assaf
  Schuster}.} \bibinfo{year}{2003}\natexlab{}.
\newblock \showarticletitle{Efficient On-the-fly Data Race Detection in
  Multithreaded {C++} Programs}.
\newblock  (\bibinfo{year}{2003}), \bibinfo{pages}{179--190}.
\newblock


\bibitem[\protect\citeauthoryear{Raman, Zhao, Sarkar, Vechev, and Yahav}{Raman
  et~al\mbox{.}}{2010}]%
        {RamanZhSa10}
\bibfield{author}{\bibinfo{person}{Raghavan Raman}, \bibinfo{person}{Jisheng
  Zhao}, \bibinfo{person}{Vivek Sarkar}, \bibinfo{person}{Martin Vechev}, {and}
  \bibinfo{person}{Eran Yahav}.} \bibinfo{year}{2010}\natexlab{}.
\newblock \showarticletitle{Efficient Data Race Detection for Async-Finish
  Parallelism}.
\newblock In \bibinfo{booktitle}{{\em Runtime Verification}},
  \bibfield{editor}{\bibinfo{person}{Howard Barringer}, \bibinfo{person}{Ylies
  Falcone}, \bibinfo{person}{Bernd Finkbeiner}, \bibinfo{person}{Klaus
  Havelund}, \bibinfo{person}{Insup Lee}, \bibinfo{person}{Gordon Pace},
  \bibinfo{person}{Grigore Rosu}, \bibinfo{person}{Oleg Sokolsky}, {and}
  \bibinfo{person}{Nikolai Tillmann}} (Eds.). \bibinfo{series}{Lecture Notes in
  Computer Science}, Vol.~\bibinfo{volume}{6418}. \bibinfo{publisher}{Springer
  Berlin / Heidelberg}, \bibinfo{pages}{368--383}.
\newblock
\showISBNx{978-3-642-16611-2}


\bibitem[\protect\citeauthoryear{Raman, Zhao, Sarkar, Vechev, and Yahav}{Raman
  et~al\mbox{.}}{2012}]%
        {RamanZhSa12}
\bibfield{author}{\bibinfo{person}{Raghavan Raman}, \bibinfo{person}{Jisheng
  Zhao}, \bibinfo{person}{Vivek Sarkar}, \bibinfo{person}{Martin Vechev}, {and}
  \bibinfo{person}{Eran Yahav}.} \bibinfo{year}{2012}\natexlab{}.
\newblock \showarticletitle{Scalable and Precise Dynamic Datarace Detection for
  Structured Parallelism}. In \bibinfo{booktitle}{{\em Proceedings of the 33rd
  ACM SIGPLAN Conference on Programming Language Design and Implementation}}
  {\em (\bibinfo{series}{PLDI '12})}. \bibinfo{pages}{531--542}.
\newblock


\bibitem[\protect\citeauthoryear{Reinders}{Reinders}{2007}]%
        {Reinders07}
\bibfield{author}{\bibinfo{person}{James Reinders}.}
  \bibinfo{year}{2007}\natexlab{}.
\newblock \bibinfo{booktitle}{{\em Intel Threading Building Blocks: Outfitting
  C++ for Multi-core Processor Parallelism}}.
\newblock \bibinfo{publisher}{O'Reilly Media, Inc.}
\newblock


\bibitem[\protect\citeauthoryear{Said, Wang, Yang, and Sakallah}{Said
  et~al\mbox{.}}{2011}]%
        {SaidWa11}
\bibfield{author}{\bibinfo{person}{Mahmoud Said}, \bibinfo{person}{Chao Wang},
  \bibinfo{person}{Zijiang Yang}, {and} \bibinfo{person}{Karem Sakallah}.}
  \bibinfo{year}{2011}\natexlab{}.
\newblock \showarticletitle{Generating Data Race Witnesses by an SMT-based
  Analysis}. In \bibinfo{booktitle}{{\em Proceedings of the Third International
  Conference on NASA Formal Methods}} {\em (\bibinfo{series}{NFM'11})}.
  \bibinfo{publisher}{Springer-Verlag}, \bibinfo{address}{Berlin, Heidelberg},
  \bibinfo{pages}{313--327}.
\newblock
\showISBNx{978-3-642-20397-8}
\showURL{%
\url{http://dl.acm.org/citation.cfm?id=1986308.1986334}}


\bibitem[\protect\citeauthoryear{Savage, Burrows, Nelson, Sobalvarro, and
  Anderson}{Savage et~al\mbox{.}}{1997}]%
        {SavageBuNe97}
\bibfield{author}{\bibinfo{person}{Stefan Savage}, \bibinfo{person}{Michael
  Burrows}, \bibinfo{person}{Greg Nelson}, \bibinfo{person}{Patrick
  Sobalvarro}, {and} \bibinfo{person}{Thomas Anderson}.}
  \bibinfo{year}{1997}\natexlab{}.
\newblock \showarticletitle{Eraser: A Dynamic Race Detector for Multi-Threaded
  Programs}. In \bibinfo{booktitle}{{\em Proceedings of the Sixteenth ACM
  Symposium on Operating Systems Principles (SOSP)}}.
\newblock


\bibitem[\protect\citeauthoryear{Serebryany and Iskhodzhanov}{Serebryany and
  Iskhodzhanov}{2009}]%
        {SerebryanyIs09}
\bibfield{author}{\bibinfo{person}{Konstantin Serebryany} {and}
  \bibinfo{person}{Timur Iskhodzhanov}.} \bibinfo{year}{2009}\natexlab{}.
\newblock \showarticletitle{ThreadSanitizer: Data Race Detection in Practice}.
  In \bibinfo{booktitle}{{\em Proceedings of the Workshop on Binary
  Instrumentation and Applications}} {\em (\bibinfo{series}{WBIA '09})}.
  \bibinfo{publisher}{ACM}, \bibinfo{address}{New York, New York},
  \bibinfo{pages}{62--71}.
\newblock


\bibitem[\protect\citeauthoryear{Smaragdakis, Evans, Sadowski, Yi, and
  Flanagan}{Smaragdakis et~al\mbox{.}}{2012}]%
        {SmaragdakisEvSa12}
\bibfield{author}{\bibinfo{person}{Yannis Smaragdakis}, \bibinfo{person}{Jacob
  Evans}, \bibinfo{person}{Caitlin Sadowski}, \bibinfo{person}{Jaeheon Yi},
  {and} \bibinfo{person}{Cormac Flanagan}.} \bibinfo{year}{2012}\natexlab{}.
\newblock \showarticletitle{Sound Predictive Race Detection in Polynomial
  Time}. In \bibinfo{booktitle}{{\em Proceedings of the 39th Annual ACM
  SIGPLAN-SIGACT Symposium on Principles of Programming Languages}} {\em
  (\bibinfo{series}{POPL '12})}. \bibinfo{publisher}{ACM},
  \bibinfo{address}{New York, NY, USA}, \bibinfo{pages}{387--400}.
\newblock
\showISBNx{978-1-4503-1083-3}
\showDOI{%
\url{https://doi.org/10.1145/2103656.2103702}}


\bibitem[\protect\citeauthoryear{Spoonhower, Blelloch, Gibbons, and
  Harper}{Spoonhower et~al\mbox{.}}{2009}]%
        {SpoonhowerBlGi09}
\bibfield{author}{\bibinfo{person}{Daniel Spoonhower}, \bibinfo{person}{Guy~E.
  Blelloch}, \bibinfo{person}{Phillip~B. Gibbons}, {and}
  \bibinfo{person}{Robert Harper}.} \bibinfo{year}{2009}\natexlab{}.
\newblock \showarticletitle{Beyond Nested Parallelism: Tight Bounds on
  Work-stealing Overheads for Parallel Futures}. In \bibinfo{booktitle}{{\em
  Proceedings of the Twenty-first Annual Symposium on Parallelism in Algorithms
  and Architectures}} {\em (\bibinfo{series}{SPAA '09})}.
  \bibinfo{publisher}{ACM}, \bibinfo{address}{Calgary, AB, Canada},
  \bibinfo{pages}{91--100}.
\newblock
\showISBNx{978-1-60558-606-9}
\showDOI{%
\url{https://doi.org/10.1145/1583991.1584019}}


\bibitem[\protect\citeauthoryear{Surendran and Sarkar}{Surendran and
  Sarkar}{2016a}]%
        {SurendranSa16b}
\bibfield{author}{\bibinfo{person}{Rishi Surendran} {and}
  \bibinfo{person}{Vivek Sarkar}.} \bibinfo{year}{2016}\natexlab{a}.
\newblock \showarticletitle{Automatic Parallelization of Pure Method Calls via
  Conditional Future Synthesis}. In \bibinfo{booktitle}{{\em Proceedings of the
  2016 ACM SIGPLAN International Conference on Object-Oriented Programming,
  Systems, Languages, and Applications}} {\em (\bibinfo{series}{OOPSLA 2016})}.
  \bibinfo{publisher}{ACM}, \bibinfo{address}{New York, NY, USA},
  \bibinfo{pages}{20--38}.
\newblock
\showISBNx{978-1-4503-4444-9}
\showDOI{%
\url{https://doi.org/10.1145/2983990.2984035}}


\bibitem[\protect\citeauthoryear{Surendran and Sarkar}{Surendran and
  Sarkar}{2016b}]%
        {SurendranSa16}
\bibfield{author}{\bibinfo{person}{Rishi Surendran} {and}
  \bibinfo{person}{Vivek Sarkar}.} \bibinfo{year}{2016}\natexlab{b}.
\newblock \showarticletitle{Brief Announcement: Dynamic Determinacy Race
  Detection for Task Parallelism with Futures}. In \bibinfo{booktitle}{{\em
  Proceedings of the 28th ACM Symposium on Parallelism in Algorithms and
  Architectures}} {\em (\bibinfo{series}{SPAA '16})}. \bibinfo{publisher}{ACM},
  \bibinfo{address}{Asilomar State Beach, CA, USA}, \bibinfo{pages}{95--97}.
\newblock


\bibitem[\protect\citeauthoryear{Surendran and Sarkar}{Surendran and
  Sarkar}{2016c}]%
        {SurendranSa16a}
\bibfield{author}{\bibinfo{person}{Rishi Surendran} {and}
  \bibinfo{person}{Vivek Sarkar}.} \bibinfo{year}{2016}\natexlab{c}.
\newblock \bibinfo{booktitle}{{\em Dynamic Determinacy Race Detection for Task
  Parallelism with Futures}}.
\newblock \bibinfo{publisher}{Springer International Publishing},
  \bibinfo{address}{Cham}, \bibinfo{pages}{368--385}.
\newblock
\showISBNx{978-3-319-46982-9}
\showDOI{%
\url{https://doi.org/10.1007/978-3-319-46982-9_23}}


\bibitem[\protect\citeauthoryear{Tarjan}{Tarjan}{1975}]%
        {Tarjan75}
\bibfield{author}{\bibinfo{person}{Robert~Endre Tarjan}.}
  \bibinfo{year}{1975}\natexlab{}.
\newblock \showarticletitle{Efficiency of a Good But Not Linear Set Union
  Algorithm}.
\newblock \bibinfo{journal}{{\it J. ACM}} \bibinfo{volume}{22},
  \bibinfo{number}{2} (\bibinfo{date}{April} \bibinfo{year}{1975}),
  \bibinfo{pages}{215--225}.
\newblock


\bibitem[\protect\citeauthoryear{Ta{\c{s}\i}rlar and Sarkar}{Ta{\c{s}\i}rlar
  and Sarkar}{2011}]%
        {TasirlarSa11}
\bibfield{author}{\bibinfo{person}{Sa{\u{g}}nak Ta{\c{s}\i}rlar} {and}
  \bibinfo{person}{Vivek Sarkar}.} \bibinfo{year}{2011}\natexlab{}.
\newblock \showarticletitle{Data-Driven Tasks and Their Implementation}. In
  \bibinfo{booktitle}{{\em Proceedings of the 2011 International Conference on
  Parallel Processing}} {\em (\bibinfo{series}{ICPP '11})}.
  \bibinfo{publisher}{IEEE Computer Society}, \bibinfo{address}{Taipei City,
  Taiwan}, \bibinfo{pages}{652--661}.
\newblock


\bibitem[\protect\citeauthoryear{Utterback, Agrawal, Fineman, and
  Lee}{Utterback et~al\mbox{.}}{2016}]%
        {UtterbackAgFi16}
\bibfield{author}{\bibinfo{person}{Robert Utterback}, \bibinfo{person}{Kunal
  Agrawal}, \bibinfo{person}{Jeremy Fineman}, {and}
  \bibinfo{person}{I-Ting~Angelina Lee}.} \bibinfo{year}{2016}\natexlab{}.
\newblock \showarticletitle{Provably Good and Practically Efficient Parallel
  Race Detection for Fork-Join Programs}. In \bibinfo{booktitle}{{\em
  Proceedings of the 28th ACM Symposium on Parallelism in Algorithms and
  Architectures}} {\em (\bibinfo{series}{SPAA '16})}. \bibinfo{publisher}{ACM},
  \bibinfo{address}{Asilomar State Beach, CA, USA}, \bibinfo{pages}{83--94}.
\newblock


\bibitem[\protect\citeauthoryear{Valdes}{Valdes}{1978}]%
        {Valdes78}
\bibfield{author}{\bibinfo{person}{Jacobo Valdes}.}
  \bibinfo{year}{1978}\natexlab{}.
\newblock {\em \bibinfo{title}{Parsing Flowcharts and Series-Parallel Graphs}}.
\newblock \bibinfo{thesistype}{Ph.D. Dissertation}. \bibinfo{school}{Stanford
  University}.
\newblock
\newblock
\shownote{STAN-CS-78-682.}


\bibitem[\protect\citeauthoryear{von Praun and Gross}{von Praun and
  Gross}{2001}]%
        {VonPraunGr01}
\bibfield{author}{\bibinfo{person}{Christoph von Praun} {and}
  \bibinfo{person}{Thomas~R. Gross}.} \bibinfo{year}{2001}\natexlab{}.
\newblock \showarticletitle{Object Race Detection}. In \bibinfo{booktitle}{{\em
  Proceedings of the 16th ACM SIGPLAN Conference on Object-oriented
  Programming, Systems, Languages, and Applications}} {\em
  (\bibinfo{series}{OOPSLA '01})}. \bibinfo{publisher}{ACM},
  \bibinfo{address}{Tampa Bay, FL, USA}, \bibinfo{pages}{70--82}.
\newblock


\bibitem[\protect\citeauthoryear{Xu, Lee, and Agrawal}{Xu
  et~al\mbox{.}}{2018}]%
        {XuLeAg18}
\bibfield{author}{\bibinfo{person}{Yifan Xu}, \bibinfo{person}{I-Ting~Angelina
  Lee}, {and} \bibinfo{person}{Kunal Agrawal}.}
  \bibinfo{year}{2018}\natexlab{}.
\newblock \showarticletitle{Efficient Parallel Determinacy Race Detection for
  Two-dimensional Dags}. In \bibinfo{booktitle}{{\em Proceedings of the 23rd
  ACM SIGPLAN Symposium on Principles and Practice of Parallel Programming}}
  {\em (\bibinfo{series}{PPoPP '18})}. \bibinfo{publisher}{ACM},
  \bibinfo{address}{Vienna, Austria}, \bibinfo{pages}{368--380}.
\newblock
\showISBNx{978-1-4503-4982-6}
\showURL{%
\url{http://doi.acm.org/10.1145/3178487.3178515}}


\bibitem[\protect\citeauthoryear{Yu, Rodeheffer, and Chen}{Yu
  et~al\mbox{.}}{2005}]%
        {YuRoCh05}
\bibfield{author}{\bibinfo{person}{Yuan Yu}, \bibinfo{person}{Tom Rodeheffer},
  {and} \bibinfo{person}{Wei Chen}.} \bibinfo{year}{2005}\natexlab{}.
\newblock \showarticletitle{RaceTrack: Efficient Detection of Data Race
  Conditions via Adaptive Tracking}. In \bibinfo{booktitle}{{\em Proceedings of
  the Twentieth ACM Symposium on Operating Systems Principles}} {\em
  (\bibinfo{series}{SOSP '05})}. \bibinfo{publisher}{ACM},
  \bibinfo{address}{New York, NY, USA}, \bibinfo{pages}{221--234}.
\newblock


\end{thebibliography}


\newcommand{\noopsort}[1]{} \newcommand{\singleletter}[1]{#1} \punt{ Uncomment
  the following lines for short conference/journal names @String{SODA = {SODA}}
  @String{JACM = {Journal of the ACM}} @String{SPAA = {SPAA}} @String{PPoPP =
  {PPoPP}} @String{PLDI = {PLDI}} @String{STOC = {STOC}} @String{FOCS = {FOCS}}
  @String{ESA = {ESA}} @String{ALP = {Colloquium on Automata, Languages, and
  Programming}} @String{SWAT = {SWAT}} @String{JALGO = {Journal of Algorithms}}
  @String{PODC = {PODC}} @String{LNCS = {LNCS}} @String{SUPERCOMP =
  {Supercomputing}} @String{ICCSE = {Israeli Conference on Computer Systems
  Engineering}} @String{CMD = {Conference on Management of Data}} }

\newpage
\appendix

\secput{generalOmitted}{Proofs from \secref{general}}

\subheading{Proof of Performance of MultiBagsPlus}

\begin{reptheorem}{thm:generalPerfBound}
\performancethmgeneral
\end{reptheorem}
\begin{proof} 
  We create three new attached sets when we encounter a \createF edge
  and 2 new sets when we encounter a \getF edge.  The only interesting
  part is when we encounter a \sync (\lirefs{syncBegin}{setAttSucc}).
  When neither of the component SP-dags have a non-SP edge 
  (\lirefs{unattachedUnionStart}{unattachedUnionEnd}) or if only one
  of them has a non-SP edge(\lirefs{oneAttachedSync}{setAttSucc}), no
  new attached sets are created.  The only case where (at most two)
  additional attached sets are created is if both subcomponents have
  non-SP edges (\lirefs{twoAttachedSync}{arcChildrenSink2}).  The total
  number of such sync nodes is $O(k)$.  Therefore, \MultiBagsPlus
  creates $O(k)$ attached sets.  Each time an attached set is created,
  it takes $O(k)$ time to insert it into $\cal R$ since \MultiBagsPlus
  maintains a transitive closure.  Other than this, each operation
  ($\proc{Make-Set}$, $\proc{Union}$ and $\proc{Find}$) into $\dssp$
  and $\dsnsp$ runs in $\iack(m,n)$. As we argued in \secref{race},
  each memory access generates a constant number of queries.  We see
  from the code that each query leads to a single $\proc{Find}$ into
  $\dssp$ and a constant number of $\proc{Find}$s into $\dsnsp$.
  Therefore, the total cost of race detection is
  $O(T_1\iack(m,n)+ k^2)$. 
\end{proof}

\subheading{Proof of Correctness of \MultiBagsPlus}

We will use the example of Figure~\ref{fig:fullExample} to illustrate
the proof.  Figure~\ref{fig:partialExample} shows the same dag and the
corresponding $\cal R$ with all the attached and unattached sets when
it has been partially executed while nodes 23 and 31 are executing,
respectively.

\begin{figure}
  \begin{minipage}{\columnwidth}
  \hspace{-5mm}
  \begin{center}
    \includegraphics[height=1.6in]{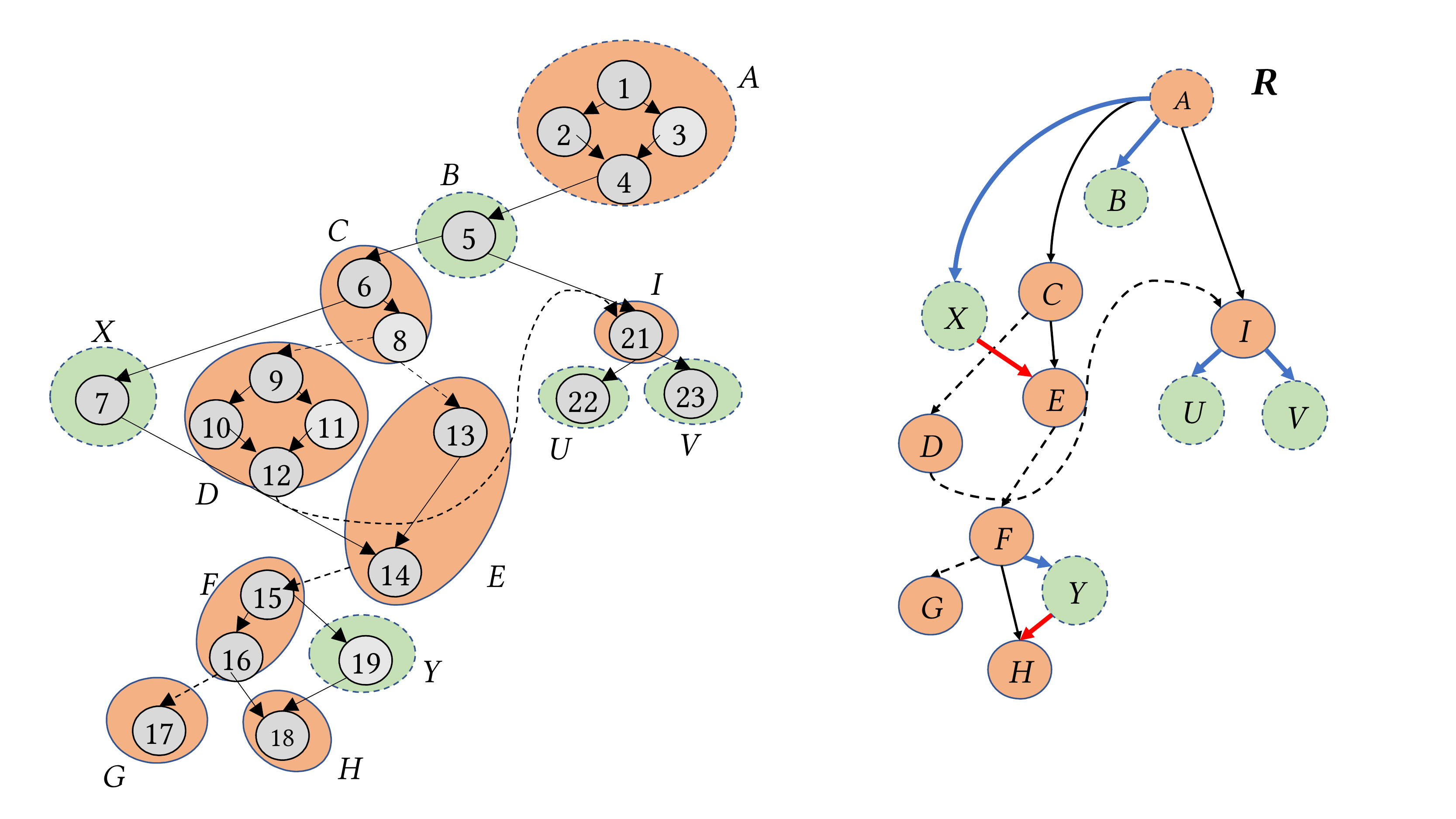}
  \end{center}
\end{minipage}
\begin{minipage}{\columnwidth}
  \flushright
  \begin{center}
    \includegraphics[height=1.6in]{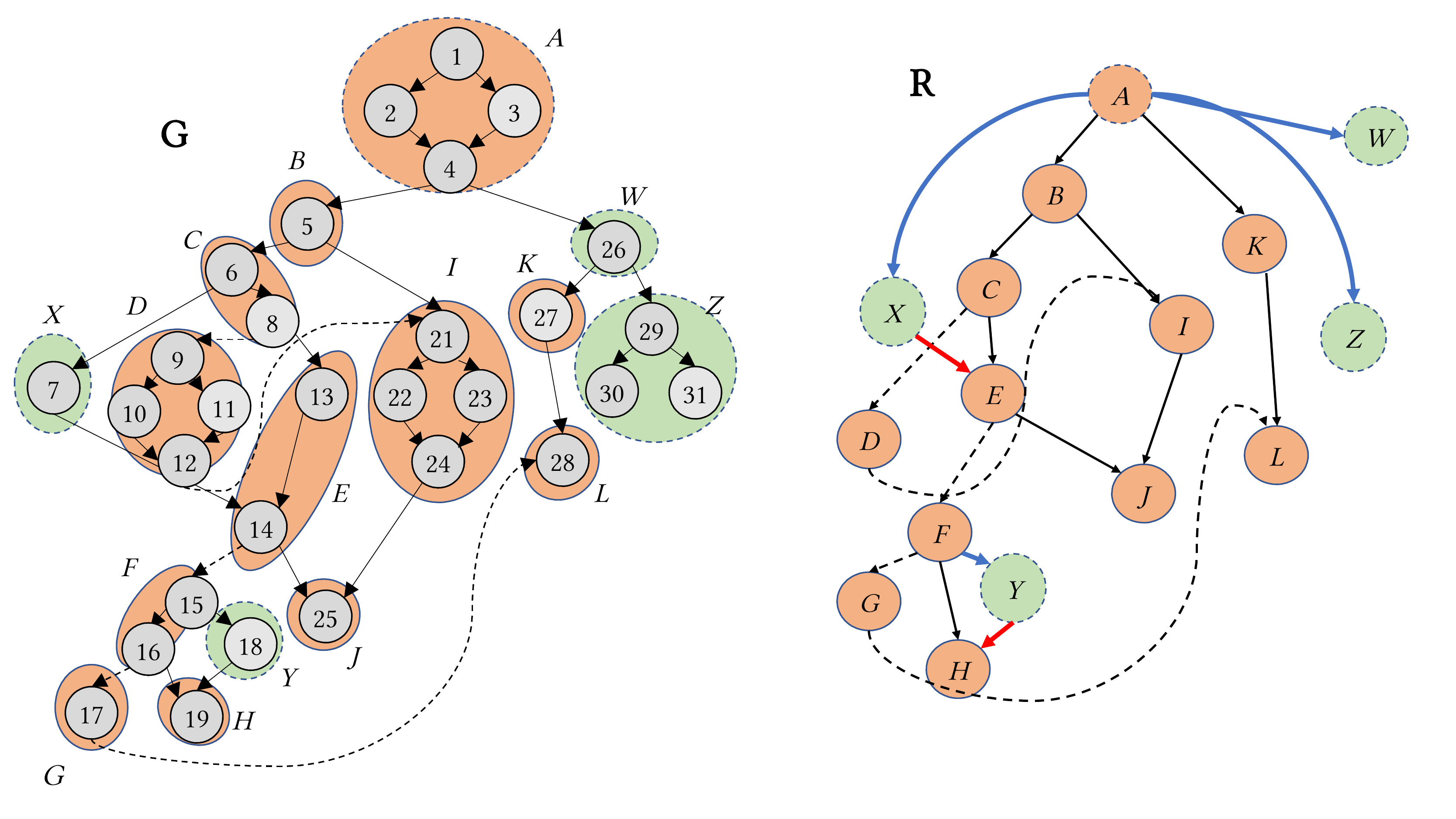}
  \end{center}
\end{minipage}
  \caption{The example from Figure~\ref{fig:fullExample} when it has
    only executed up to node 23 (top) and node 31 (bottom).  Note that
    sets $U$ and $V$ do not yet have attached successors on the left.
  }\figlabel{partialExample}
\end{figure}

We will redefine the terminology a little bit differently than
Section~\ref{sec:structured}.  A node $u$ is a \defn{spawn
  predecessor} of a node $v$ if there is a path from $u$ to $v$ which
consists of only spawn, create and continue edges.  A node $u$ is a
\defn{join predecessor} of $v$ if there is a path from $u$ to $v$ that
consists of only join and continue edges.  Notice the asymmetry here
--- a creator node is the corresponding future's spawn successor, but
the future is not the getter node's join predecessor.  This mimics the
actions of the algorithm on $\dssp$, since \spawn and \createF behave
identically while \sync and \getF do not.  

We can now define operating function and confluence in the manner
identical to Section~\ref{sec:structured} and it should be clear that
Property~\ref{prop:eager} and Lemma~\ref{lem:SPStuff} still hold.  

First consider the first part of the query, where we just check if $u$
is in the $S$-bag.  As shown in \lireftwo{dsspspawn}{dsspcreate}, when
a function $F$ calls either $\spawn(G)$ of $\createF(G)$, the
algorithm exactly mimics \MultiBags, simply creating an $S$-bag for
$G$ containing the first node of $G$.  Similarly, on function $G$'s
return, $S_G$ becomes $P_G$.  Finally, when a function $F$ calls \sync
on function $G$, it mimics \getF in structured future and $P_G$ is
unioned in $S_F$.  The only difference is that nothing happens on
$\getF$ to $\dssp$. Therefore:

\newcommand{\dssplemma}{ Consider the
  currently executing strand $v$ and a previously executed strand $u$.
  The following is true: (a) If $u$ is in an $S$ bag, then
  $u \prec v$; further more, there is a path from $u$ to $v$
  consisting of only spawn, create, join and continue edges.  (b) If
  there is a path from $u$ to $v$ consisting of only spawn, create,
  join and continue edges (no get edges), then $u$ is in an $S$ bag.}

\begin{lemma}
\dssplemma
\lemlabel{dssp}
\end{lemma}
\begin{proof}
  Since moves to $S$-bags happen in a more restricted setting here
  than in \secref{structured}, statement (a) follows somewhat
  intuitively from Section~\ref{sec:structured}.  More formally, from
  Lemma~\ref{lem:SPStuff}, if $u$ in an $S$ bag, then either the
  function containing $u$ is active, or $u$'s operating function is
  active.  Therefore, either $u$ or some join successor of $u$ is part
  of an active function --- therefore, $u$ is sequentially before $v$
  by Property~\ref{prop:eager}.

  For (b), say
$X=\spDag(u)$ and let $w$ be the last node in $X$ on the path from $u$
to $v$ (if $v$ is in $X$, then $v=w$).  Since the path from $u$ to $w$
does not contain any \getF edges, this path can not go through any
\createF edges either (since a \createF edge will take execution out
of $X$ and one would need a \getF edge to come back).  Therefore,
since SP dags are a special case of structure futures,
and $u$ is sequentially before $w$, $u$ must have been in some $S$ bag
when $w$ executed (from Theorem~\ref{thm:correctness}).  After $w$
executed, no strand of $X$ was executed, therefore $u$ can not have
moved to a $P$ bag since then.
\end{proof}

We now consider the case when $\dssp$ returns false and use $\dsnsp$.
Here, we must prove that if $u$ is in a $P$ bag when $v$ executes
(otherwise, the query returned), then $u \prec_\full v$ iff there is a
path in $\cal R$ from $\attsucc{u}$ to $\attpred{v}$.  We have to make
a few observations in order to see this.

The following two lemmas state some structural properties of attached
and unattached sets and can be proven by inducting on set unions done
by the algorithm.

\newcommand{\unattachedStructProps}{Each unattached set consists of
  the nodes belonging to a maximal series-parallel subdag $Q$ such
  that (1) all nodes in the subdag have been executed, (2) there are
  no \createF or \getF edges in $\gfull$ incident on nodes in $Q$, (3) there are
  at most two arcs in $\gfull$ incident on nodes in $Q$---one directed
  towards $Q$'s source, and one directed out from $Q$'s sink.
}
\begin{lemma}
\unattachedStructProps
\lemlabel{unattachedStruct}
\end{lemma}
\begin{proof}
  Induct on the construction.  When we create an unattached set, it
  consists of a single node.  When we merge sets without marking them
  attached, they are always complete series parallel dags
  (\lirefs{unattachedUnionStart}{unattachedUnionEnd}) with no incident non-SP edges.  
\end{proof}

\newcommand{\attachedStructProps}{ The nodes in each attached set
  induce a series-parallel subdag $Q$ such that if we contract all
  maximal series parallel dags within it, it forms a chain.(1) if there
  is a \getF or a \createF edge directed towards a node in $Q$, then it
  is directed towards the source, (2) if there is a \getF or \createF
  edge directed from a node in $Q$, then it originates on the sink.
  Unlike the unattached sets, attached sets do not necessarily match
  the series-parallel decomposition of $\gfull$---an attached set may
  have many incoming or outgoing edges.   }

\begin{lemma}
\attachedStructProps
  \lemlabel{attachedStruct}
\end{lemma}
\begin{proof}
  We induct on the growth of sets.  The sets always start by
  containing single nodes.  We union sets at two places: On
  \lirefs{unattachedUnionStart}{unattachedUnionEnd}, we union an
  entire series-parallel subdag together.  In this case, the sets
  containing $t_1$ and $t_2$ are unattached and have no incident
  non-SP edges by Lemma~\ref{lem:unattachedStruct}.  Therefore, only
  the set containing $f$ may be attached. By induction, it can only
  have an incident non-SP edge at its source node.  Before the unions,
  Its sync node is $f$ and it clearly does not have an outgoing non-SP
  arc since it is a spawn node, not a \createF node.  

  On \liref{singleUnionSrc}, we union the set which contains a spawn
  node $f$ with a set containing one of spawn's successors $s_a$.  By
  induction, the set containing $f$ only has an incident non-SP edge 
  towards its source.  Similarly, the set containing $s_a$ can only
  have a non-SP edge directed away from its sink node (again $s_a$ is
  its source node and it does not have a non-SP edge directed towards
  it since its previous instruction $f$ is not a \getF or \sync call).

  Similarly, on \liref{singleUnionSync}, we union a sync node with a
  set containing sync nodes successor.  Again, we induct in a similar manner.
\end{proof}

We can see that the above lemmas are true for examples shown in
Figures~\ref{fig:fullExample} and~\ref{fig:partialExample}.  From the
above two lemmas, we can see that if $u$ and $v$ are in the same
attached or unattached set, and $u \prec_\full v$ then $u\prec_{SP} v$
since $u$ and $v$ are in the same SP-dag and there are no incoming or
outgoing non-SP edges to nodes within the dag except at the source and
sink.  Therefore, $\dssp$ will answer the query between them
correctly.

The following key lemma says that the relationship between nodes in
two different attached sets is always correctly represented in
$\cal R$.  \newcommand{\RcorrectLemma}{Consider nodes $u$ and $v$,
  where $A_u = \find{\dsnsp, u}$ and $A_v = \find{\dsnsp, v}$ are
  distinct attached sets.  Then $u \prec_\full v$ iff
  $\find{\dsnsp,u} \prec_{\cal R} \find{\dsnsp, v}$.}

\begin{lemma}
\RcorrectLemma
\lemlabel{RCorrect}
\end{lemma}

The proof of the lemma is complicated and depends on many structural
properties of the sets; however, there are two main intuitions.
First, if we only consider only attached sets containing nodes of a
single series parallel dag $D$ (ignoring all \createF and \getF edges)
then the attached sets form the same series parallel relationship as
$D$.  This can be seen from in Figure~\ref{fig:fullExample} where
$A, B, C, E, I, J, K, L, M$ and $N$ form a series parallel dag and
induce the same dependencies as the original dag.  Therefore, if the
path from $u$ to $v$ doesn't contain any \createF or \getF edges, then
we get the correct relationship.

Second, say there is a path $p$ from $u$ to $v$ which contains
\createF and \getF edges.  Walk along the path from $u$ until we
encounter the first such edge.  Let $w$ the source of this edge and
$x$ be the destination.  If both $u$ and $w$ are in attached sets,
then there is a path from $u$ to $w$ in $\cal R$ (from the previous
paragraph).  In addition, there is an edge from $w$ to $x$ in $\cal R$
(since $w$ is the creator node and $x$ is the first node of a future
and we explicitly add this edge on \liref{arcCreatorFut}).  Therefore,
there is a path from $u$ to $x$ in $\cal R$.  We can then find the
next \createF or \getF edge on this path and continue with the
induction.  Again, we can see this in Figure~\ref{fig:fullExample}
where there is a path from $C$ to $I$ via $D$ and also a path from $C$
to $L$ via $E, F$ and $G$.

In order to prove Lemma~\ref{lem:RCorrect} formally, we will define
two kinds of attached sets.  Intuitively, \defn{prefix-complete sets}
are those where the first node added to the set is the source node of
the set and the set can grow in the forward direction; and
\defn{suffix-complete sets} are those where the first node added is
the sink node of the set and the set can grow in the backward
direction by unions with unattached sets.

We first understand the structure of prefix-complete attached sets.
Consider an attached set $A$ where the first node added to $u$.  $A$
is prefix complete if $u$ is either (1) the strand immediately after a
\getF node (\liref{AttachifyGet2}); or (2) $u$ the first node in the
continuation of a creator node (\liref{AttachifyCreate3}); or (3) $u$
is the first node of a future function (\liref{AttachifyCreate2}); or
(4) a sync node where both its subdags that are joining are attached
(\liref{joinAttachify}).

\begin{lemma}
If $u$ is the first node added to a prefix-complete attached set $A$, then
for all other nodes $v \in A$, we have $u \prec v$.  
  In addition, consider a node $v \in A$ where $A$ is a prefix-complete attached
  set.  (1) If there is an edge $(v,w)$ in $\gfull$ such that $w$ is not in
  $A$, then $v$ must be the sink node of $A$.  (2) If there
  is any edge $(w,v)$ in $\gfull$ where $w$ is not in $A$, then $v$ is
  either the source or $w$ is in an unattached set.
  \lemlabel{prefixAttachedSetStructure}
\end{lemma}
\begin{proof}
  The first statement is clear by construction. Any attached set
  constructed in one of the ways described above will never union with
  an unattached set that has nodes that precede $u$.  In particular,
  we can induct on the growth of a prefix-complete set to show all the
  properties described above.  
\end{proof}

Consider an attached set $A$ where the first node added to $u$.  $A$
is suffix-complete if $u$ is either (1) the strand ending with a
\createF instruction
(\liref{AttachifyCreate1}); or (2) the node immediately preceding the
getter node in the SP dag of the getter node
(\liref{AttachifyGet1}); or (3) the spawn node where both
spawned subdags are attached (\liref{spawnAttachify}).

\begin{lemma}
If $u$ is the first node added to a prefix-complete attached set $A$, then
for all other nodes $v \in A$, we have $v \prec u$.  
Consider a node $v \in A$; if there is an edge $(w,v)$ in $\gfull$
such that $w$ is not in $A$, then $v$ must be the source node of $A$.
In addition, if there is any edge $(v,w)$ in $\gfull$ where $w$ is not
in $A$, then $v$ is either the sink or $w$ is in an unattached set.  
\lemlabel{suffixAttachedStructure}
\end{lemma}
\begin{proof}
Again, we can induct on the growth of sets created by the above
methods.
\end{proof}


The following is a surprising lemma.  Basically, a suffix-complete
attached set is never an unattached set's attached successor or
predecessor.  
\begin{lemma}
If an attached set $A$ is an attached predecessor or an attached
successor of an unattached set $U$, then $A$ must be prefix-complete.
\lemlabel{attachedSuccPrecPr}
\end{lemma}
\begin{proof}
  A set $A$ is set as an attached successor (on \liref{setAttSucc})
  when a join node $j$ unions with $A$.  In this case, $A$ is clearly
  growing in the forward direction and must be prefix complete.  The
  first node of the entire computation, is by definition,
  prefix-complete.  After this, we can see that attached predecessor
  is always prefix-complete by inducting on the execution.
\end{proof}

We now consider nodes in unattached sets and argue that they have the
correct relationship with their predecessors.  If $u$ is in an
unattached set, then we use $\attpred{u}$ as a proxy for $u$ when we
do the query.  The following lemma argues that it is always correct to
use this proxy.

\newcommand{\attPredLemma}{At any point during the execution, for an
  unattached set $U$, if $A = \attpred{U}$, then for all $u \in A$ and
  $v \in U$, we have $u \prec_{SP} v$. In addition, there is no
  incoming \createF or \getF edge on any node in the path from $u$ to
  $v$ (not including $u$).}
\begin{lemma}
\attPredLemma
\lemlabel{attPred}
\end{lemma}
\begin{proof}
  First, recall that $A$ must be a prefix-complete set
  (Lemma~\ref{lem:attachedSuccPrecPr}), and the only edges leaving in
  $\gfull$ can be from its sink node.  Therefore, if any node $u$ in
  $A$ precedes $v$, then they all must.  

  The fact that some node $v$ in $A$ must precede $u$ can be seen by
  induction.  If $v$ has only one immediate predecessor
  $w \in \spDag(v)$, then $v$'s attached predecessor is set as $w$'s
  attached predecessor (\lireftwo{setAttPredSpawn}{setAttPredCont}).
  If $v$ has two immediate predecessors ($v$ is a join node), then $v$
  is either in an attached set
  (\lireftwo{joinAttachify}{singleUnionSync}) or it is in the same set
  with both its predecessors.
\end{proof}

We can also see this from our examples.  The example of set $X$ is
particularly interesting.  Note that its attached predecessor is not
$C$ --- this is because $C$ was not an attached set when $X$ executed
--- it only became an attached set later.  However, note that all
nodes in $A$ are in fact before all nodes in $X$ and there is no
intervening incident non-SP edge.  Similar observations can be made
for set $Z$ whose attached predecessor is $A$ instead of $K$ since
node 26 was not in $K$ when $Z$ executed (as seen in
Figure~\ref{fig:partialExample} (right)).  However, it is still correct for
$Z$'s attached predecessor to be $A$ since every node in $A$ precedes every
node in $Z$ and there is no non-SP edge on any path from a node in $A$ to a
node in $Z$.  

We can now prove that $\cal R$ has the correct relationships between
attached sets.  

\begin{replemma}{lem:RCorrect}
\RcorrectLemma
\end{replemma}
\begin{proof}
  We first argue that if $\find{\dsnsp,u} \prec_R \find{\dsnsp, v}$,
  then $u \prec_\full v$.  We can induct on order in which edges are
  added in $\cal R$.  Let $A_u = \find{\dsnsp,u}$ and
  $A_v = \find{\dsnsp, v}$.  If $v$ is the first node added to $A_v$
  via $\proc{Attachify}(v)$, then all the incoming edges to $A_v$ are
  from 
  $\attpred{v}$ --- therefore, from Lemma~\ref{lem:attPred}, all paths
  into $A_v$ are correct.  If $A_v$ is directly created
  (\lirefthree{AttachifyCreate3}{AttachifyGet2}{joinAttachify}) then
  we explicitly only add the edges into $A_v$ in $\cal R$ which are correct
  in $\gfull$.  Also, at this point, either $A_v$ has no outgoing
  edges or the correct outgoing edges are explicitly added. 

  When we execute a \sync, we may union unattached sets into attached
  sets, and we must ensure that the property still holds.  (1) When
  both subdags of an SP dag are unattached, they are both unioned into
  the set containing the source and the join is also unioned into the
  same set (\lirefs{unattachedUnionStart}{unattachedUnionEnd}).  In this case, the
  only things that can precede these unattached sets also precede the
  source.  (2) When both subdags are attached, no unions happen.  (3)
  When one subdag is attached, the join node $j$ is added to the set
  that contains one of the predecessor of $j$ (\liref{singleUnionSync})
  --- therefore, it must be the case that a node that precedes the
  predecessor of the join node also precedes $j$.  In addition, no
  node succeeds $j$ yet, so the other direction is trivial.  The
  source node union (\liref{singleUnionSrc}) is the interesting case.
  Here $s$ is unioned into a suffix-complete attached set $s_a$.
  Therefore, $s_a$ has only one incoming edge and by definition, it is
  from $s$.  Therefore, anything that precedes nodes in $s_a$ must
  also precede $s$ and anything that succeeds nodes in $s_a$ must also
  succeed $s$.  

  Now we argue that if $u \prec_\full v$, then either $u$ and $v$ are
  in the same attached set or the set containing $u$ precedes the set
  containing $v$ in $\cal R$.  First observation is that $\cal R$ is always a
  connected dag --- this is easy to see since whenever we add a node
  in $\cal R$, we also add an arc to it.  

  Second, if $u$ and $v$ are in the same SP-dag, and they are not in
  the same attached set, then there is a path from $A_u$ to $A_v$ in
  $\cal R$.  We can again see this by induction on composition of SP
  dags.  In the base case, $u$ and $v$ are in different attached sets,
  we immediately add an edge from $u$ to $v$.  After this, assuming
  there is always an edge from the source to sink in the smaller sp
  dag, we always either merge the entire sp-dag into the same set
  (\lirefs{unattachedUnionStart}{unattachedUnionEnd}), or add edges
  from source to both sub dag sources
  (\lirefs{arcSourceChildren1}{arcSourceChildren2})
  and from subdag sinks to the sink
  (\lirefs{arcChildrenSink1}{arcChildrenSink2}) or we
  merge the source into one of the subdag sources
  (\liref{singleUnionSrc}) and sink to one of the subdag sinks
  (\liref{singleUnionSync}).

  If $u$ and $v$ are in different SP dags, we induct on the path from
  $u$ to $v$.  Let $w$ be the last node in $\spDag(u)$ and $x$ be the
  node immediately after $w$ in the path from $u$ to $v$.  Both $w$
  and $x$ are in attached sets and there is a path from $u$ to $w$ in
  $\cal R$ (from the previous paragraph) and an edge from $w$ to $x$ in $\cal R$
  (since $w$ is the creator node and $x$ is the first node of a future
  and we explicitly add this edge on \liref{arcCreatorFut}).  Therefore, there is a
  path from $u$ to $x$.  We can then induct on this path and keep
  moving forward until we get to the dag containing $v$.  
\end{proof}

Finally, we must make claims about attached successors of unattached sets.  We first show
a structural property of unattached sets. 

\begin{lemma}
For a node $u$, let $C_u$ be the closest completed series-parallel dag
which is a parallel composition and whose sink node has already
executed.    If $u$ belongs to an
unattached set with no attached successor, then all nodes of $C_u$
belong to this same unattached set.  
\lemlabel{closestUnattached}
\end{lemma}
\begin{proof}
  When a sink node of a parallel composition following an unattached
  set executes, the unattached set either gets an attached successor
  \liref{setAttSucc}, or all the nodes of the parallel composition
  are unioned into the same
  set~\lirefs{unattachedUnionStart}{unattachedUnionEnd}.
\end{proof}

The following
lemma claims that if a node $u$ has an attached successor, then there
is a path from $u$ to the last node of the attached successor.  

\newcommand{\attSuccLemma}{ Consider a node $u$ where
  $U_u = \find{\dsnsp, u}$ is unattached.  If $U_u$ has an attached
  successor $A$, then $u \prec_{SP} v$ where $v$ is the current sink
  of this attached successor $A$.  In addition, consider any node
  $w \not \in U_u$.  If $u \prec_\full w$ then $v \prec_\full w$.
  Finally, we have $u\prec w$ for any node $w$ added to $A$ after it
  becomes $U_u$'s attached successor. }
\begin{lemma}
\attSuccLemma
 \lemlabel{attSuccSink}
\end{lemma}
\begin{proof}
  Attached successor is set on \liref{setAttSucc} where the attached
  successor always contains the sync node $j$ following the $U_u$.
  Since an unattached set is a complete series-parallel dag with no
  incident \createF or \getF edges (Lemma~\ref{lem:unattachedStruct}),
  any path from $u$ must go through this node $j$ from the property of
  series-parallel dags.  At this point $j$ is the current sink of $A$.
  Any nodes $w$ subsequently added to $A$ must have the property that
  $j \prec w$ from the construction of attached
  sets(Lemma~\ref{lem:attachedStruct}).  Therefore, the property
  remains true by induction. 
\end{proof}

The important subtlety here is that not all nodes of $U_u$'s attached
successor have to be after $u$ --- consider the set $Z$ in our example
with attached successor $L$.  Node 28 does not follow $Z$, but $L$
becomes $Z$'s attached successor only after node 33 is added to $L$.
(Notice that $L$ is not $Z$'s attached successor in
Figure~\ref{fig:partialExample} (right) since 33 has not yet executed.)

Now we can argue about the correctness of the query
$\proc{Query}(u,v)$ where both $u$ and $v$ may be parts of unattached
sets.  Say $u$ has an attached successor $A_1$ and $v$ has an attached
predecessor $A_2$ and there is no path from $u$ to $v$ that contains
only series-parallel edges (otherwise, the first part of the query
will give the correct answer), but there is a path $p$ from $u$ to $v$
containing \createF and \getF edges.  This path must go through
the last node of $A_1$ (Lemma~\ref{lem:attSuccSink}.  In addition, since there
are no incident non-SP edge between nodes in $A_2$ and $v$
(Lemma~\ref{lem:attPred}), this path must also go through $A_2$. Therefore, it
is sufficient to check the relationship between $A_1$ and $A_2$ to check the
relationship between $u$ and $v$.  In a similar manner, we can also show that
we get the correct answer when only one of $u$ and $v$ are in unattached sets.  

This final lemma handles the case where $u$ does not have an attached
successor.  This is where we utilize the condition that the program is
executing in depth-first eager order and that $v$ is the currently
executing node.  The intuition for this lemma is as follows: Since $u$ is in an
unattached set, no node in this unattached set has any outgoing non-SP
edges.  In addition, the nearest join after this set hasn't executed
since it doesn't have an attached successor.  Therefore, all no node
$v$ where $u \prec v$ can be executing right now since all such nodes
are in the same $P$ bag as $u$.  \footnote{We also use this fact when we are querying
  the $\dssp$ data structure since the correctness of
  Lemma~\ref{lem:dssp} also depends on eager execution.}

\newcommand{\noAttSuccLem}{ Consider a node $u$ where
  $U_u = \find{\dsnsp, u}$ is unattached and $U_u$ does not have an
  attached successor.  If $u$ is an $P$ bag, $u$ is in parallel with
  the currently executing node $v$.}

\begin{lemma}
\noAttSuccLem
\lemlabel{noAttSucc}
\end{lemma}
\begin{proof}
  Say $u$ was a strand of function $F$ and its operating function is
  $G$ ($F$ and $G$ could be the same).  Therefore, $u$ is in $G$'s $P$
  bag and $G$ is not active, but has returned.  Moreover, $G$ has
  not synced with its parent function.  In addition $G$'s last strand $w$
  is in the same unattached set as $u$
  (Lemma~\ref{lem:closestUnattached}).  Since no node of this
  unattached set has outgoing non-SP edges, no node that is
  sequentially after $u$ has any incident non-SP edges.  Therefore,
  from Lemma~\ref{lem:dssp}, there can not be any \getF edges on any
  path after $u$ at this point.  Therefore $u$ is in parallel with
  $v$.  
\end{proof}

We can now prove the main theorem by
combining
Lemmas~\ref{lem:dssp},~\ref{lem:RCorrect},~\ref{lem:attPred},~\ref{lem:attSuccSink},
and~\ref{lem:noAttSucc}.

\begin{reptheorem}{thm:unstructCorrect}
\correctnessthmgeneral
\end{reptheorem}

If the path from $u$ to $v$ has no \getF edge, then
Lemmas~\ref{lem:dssp} applies.  In addition,
Lemma~\ref{lem:RCorrect} argues that the second part of the query
(\lirefs{part2Start}{part2End}) answers all questions correctly
between two nodes in attached sets. Finally,
Lemmas~\ref{lem:attPred},~\ref{lem:attSuccSink},
and~\ref{lem:noAttSucc} show that using attached predecessors and
successors for nodes in unattached sets gives the correct answer when
the first part of the query (using $\dssp$) returns false.

%
%
%
%
%






\section{Artifact Appendix}

\subsection{Abstract}

This artifact contains source code for the compiler, runtime system,
and benchmarks used in the PPoPP 2019 paper {\bfseries Efficient Race
  Detection with Futures}, plus shell scripts that compile everything
and run the benchmarks. The hardware requirements are any modern
multicore CPU, while the software requirements include a relatively
recent Linux distribution (tested on Ubuntu 16.04), the
\texttt{datamash} package, and the GNU \texttt{gold} linker. To
validate the results, run the test scripts and compare the results to
figures 6, 7, and 8 in the paper.

\subsection{Artifact check-list (meta-information)}

\newcommand{\rtm}{\textsuperscript{\textregistered}}


{\small
\begin{itemize}
  \item {\bf Program: } C/C++ code.
  \item {\bf Compilation: } Modified fork of \texttt{clang++} with
    \texttt{-O3 -flto} flags. To fully reproduce the reproduce the
    results, we recommend installing the GNU {\texttt gold} linker as
    \texttt{ld}.
  \item {\bf Data set: } The \texttt{dedup} benchmark uses publicly
    available data sets. Scripts in the repository will download and
    setup all data sets.
  \item {\bf Run-time environment: } Tested on Ubuntu 16.04, but
    expected to work on any modern Linux. 
  \item {\bf Hardware: } Any modern multicore CPU; tested on an
    Intel\rtm Xeon\rtm \ CPU E5-2665 with hyperthreading
    disabled. Enabling hyperthreading may change results.
  \item {\bf Metrics: } Runtime (in seconds).
  \item {\bf Output: } Runtime and standard deviation for all
    benchmarks, each run with 12 configurations which determine what
    kind of futures and which race detection algorithm are used and
    what level of instrumentation/race detection is turned on ---
    baseline, reachability only, reachability + memory
    instrumentation, and full race detection.
  \item {\bf How much disk space required (approximately)?: } 13GB.
  \item {\bf How much time is needed to prepare workflow
    (approximately)?: } 1.5 hours.
  \item {\bf How much time is needed to complete experiments (approximately)?: 4 hours.}
  \item {\bf Publicly available?: } Yes
  \item {\bf Code/data licenses (if publicly available)?: } MIT.
\end{itemize}

\subsection{Description}

\subsubsection{How delivered} 

The project is available on Gitlab at
\url{https://gitlab.com/wustl-pctg-pub/futurerd2.git}.

\subsubsection{Hardware dependencies}

Any modern multicore CPU. It was tested on an Intel\rtm Xeon\rtm CPU
E5-2665.

\subsubsection{Software dependencies}

The project was tested on Ubuntu 16.04, but it is expected to run
correctly in other Linux distributions. To fully reproduce the
results, link-time optimization should be used (\texttt{-flto}) with
the GNU \texttt{gold} linker installed as \texttt{ld}. On our system
we make \texttt{/usr/bin/ld} a shell script that forwards its
arguments to \texttt{gold} whenever the \texttt{USE\_GOLD} environment
variable is set and the original \texttt{ld} otherwise.

The benchmark script requires GNU {\texttt datamash}, which can be
installed using {\texttt apt-get} in Ubuntu 14+ or can be obtained
from \url{https://www.gnu.org/software/datamash}. Bash 4+ should be
used to run the scripts.

\subsubsection{Data sets}

All required datasets are downloaded by scripts included in the
distribution.

\subsection{Installation}

The {\texttt setup.sh} script in the project repository will build our
modified compiler, the modified Cilk Plus runtime, and all the
benchmarks.

\subsection{Experiment workflow}

\begin{itemize}
\item Clone the source code to your machine:

  \lstset{basicstyle=\footnotesize, basewidth=0.5em}
  \begin{lstlisting}
$ git clone
>  https://gitlab.com/wustl-pctg-pub/futurerd2.git
$ cd futurerd2    
  \end{lstlisting}
\item Install GNU gold as your linker. Modern versions of the GNU
  {\texttt binutils} package include {\texttt gold}, though for our
  purposes the system {\texttt ld} should point to {\texttt
    gold}. Installing {\texttt gold} also installs a header called
  {\texttt plugin-api.h}, usually in either {\texttt /usr/include} or
  {\texttt /usr/local/include}. Find this file and replace the
  {\texttt BIN\-UTILS\_PLU\-GIN\_DIR} variable in {\texttt build-llvm-linux.sh}
  with this path.
\item Install other software dependencies. In Ubuntu 14+, this is as simple as

  \begin{lstlisting}
$ sudo apt-get install datamash zlib1g zlib1g-dev openssl
  \end{lstlisting} 

  and making sure you have Bash 4+.

\item Build the necessary components. The {\tt setup.sh} script will
  build the compiler and download and unpack the necessary data sets.

\item Run the benchmark script (\texttt{bench/run.sh}). The script
  compiles the runtime library and race detection library, and
  compiles and runs each configuration of each benchmark. Tuning
  parameters can be found in \texttt{bench/time.sh} (which the
  \texttt{run.sh} script uses) --- feel free to examine the script and
  change parameters, such as the number of iterations for each
  benchmark.
  \begin{lstlisting}
$ cd bench
$ ./run.sh    
  \end{lstlisting}

\item Observe the results. Once completed, full results can be found
  in the files \texttt{times.ss.csv} (benchmarks used \MultiBags race
  detection algorithm with structured futures), \texttt{times.ns.csv}
  (benchmarks used \MultiBagsPlus algorithm with structured futures),
  and \texttt{times.nn.csv} (benchmarks used \MultiBagsPlus algorithm
  with general futures).

\end{itemize}

\subsection{Evaluation and expected result} 

Although absolute times will differ on your machine, you should see
similar relative overhead for the benchmarks. Compare the results to
figures 6, 7, and 8 in the paper.



\subsection{Notes}

Please send feedback or file issues at our gitlab repository
(\url{https://gitlab.com/wustl-pctg-pub/futurerd2}).

\subsection{Methodology}

Submission, reviewing and badging methodology:

\begin{itemize}
  \item \url{http://cTuning.org/ae/submission-20180713.html}
  \item \url{http://cTuning.org/ae/reviewing-20180713.html}
  \item \url{https://www.acm.org/publications/policies/artifact-review-badging}
\end{itemize}



\end{document}